\title{Asymptotically Optimal Regenerating Codes Over Any Field}

\documentclass[12pt]{article}

\usepackage[affil-it]{authblk}

\author{Netanel Raviv}
\affil{\normalsize Computer Science Department, \\Technion -- Israel Institute of Technology,\\Haifa 3200003, Israel}
\date{}

\renewcommand{\thefootnote}{\arabic{footnote}}

\usepackage{pgf,tikz}
\usetikzlibrary{arrows}
\usepackage[margin=0.7in]{geometry}
\usepackage{amssymb}
\usepackage{amsmath}
\usepackage{amsthm}

\usepackage{color}
\usepackage[hyperfootnotes=false]{hyperref}
\usepackage{mathtools}
\usepackage{multirow}
\usepackage{hhline}
\usepackage[sort,numbers]{natbib}
\usepackage{tabularx}

\usepackage{framed}

\usepackage{stmaryrd}

\DeclarePairedDelimiter{\ceil}{\lceil}{\rceil}

\DeclarePairedDelimiter{\subm}{\llbracket}{\rrbracket}

\newtheorem{theorem}{Theorem}
\newtheorem{definition}{Definition}

\newtheorem{lemma}{Lemma}

\newtheorem{remark}{Remark}

\newcommand{\cl}[1]{\mathcal{#1}}

\newcommand{\bF}{\mathbb{F}}

\newcommand{\bZ}{\mathbb{Z}}

\newcommand{\cC}{\mathcal{C}}

\newcommand{\cM}{\mathcal{M}}
\newcommand{\cN}{\mathcal{N}}

\newcommand{\cR}{\mathcal{R}}

\newcommand{\cU}{\mathcal{U}}
\newcommand{\cV}{\mathcal{V}}

\newcommand{\grsmn}[3]{\cl{G}_{#1}\left(#2,#3\right)}

\DeclareMathOperator{\rank}{rank}

\DeclareMathOperator{\image}{Im}
\DeclareMathOperator{\vect}{vec}
\DeclareMathOperator{\ord}{ord}

\makeatletter
\def\namedlabel#1#2{\begingroup
	\def\@currentlabel{#2}%
	\label{#1}\endgroup
}
\makeatother

\newcommand\blfootnote[1]{%
	\begingroup
	\renewcommand\thefootnote{}\footnote{#1}%
	\addtocounter{footnote}{-1}
	\endgroup
}

\begin{document}
\maketitle

\begin{abstract}
The study of regenerating codes has advanced tremendously in recent years. However, most known constructions require large field size, and hence may be hard to implement in practice. By using notions from the theory of extension fields, we obtain two explicit constructions of regenerating codes. These codes approach the cut-set bound as the reconstruction degree increases, and may be realized over any given field if the file size is large enough. Since distributed storage systems are the main purpose of regenerating codes, this file size restriction is trivially satisfied in most conceivable scenarios. The first construction attains the cut-set bound at the MBR point asymptotically for all parameters, whereas the second one attains the cut-set bound at the MSR point asymptotically for low-rate parameters. 
\blfootnote{This work was done while Netanel Raviv was a visiting student at the University of Toronto, under the supervision of Prof.~Frank Kschischang. It is a part of his Ph.D. thesis performed at the Technion, under the supervision of Prof.~Tuvi Etzion. \emph{e-mail}: \texttt{netanel.raviv@gmail.com}.}
\end{abstract}

\section{Introduction}\label{section:Introduction}
Since the emergence of cloud storage platforms, distributed storage systems are ubiquitous. As classic erasure correction codes fail to scale with the exponential growth of data, regenerating codes were proposed~\cite{CutSetBound}. 

A regenerating code is described by the parameters $(n,k,d,B,q,\alpha,\beta)$, where $k\le d\le n-1$ and $\beta \le \alpha$. The file $x\in\bF_q^B$ is to be stored on $n$ \textit{storage nodes}. The \textit{reconstruction degree}~$k$ is the number of nodes required to restore~$x$, a process which is called \textit{reconstruction}, and carried out by a \textit{data collector}. The \textit{repair degree}~$d$ is the number of \textit{helper nodes} which are required to restore a lost node, a process which is called \textit{repair}, and carried out by a \textit{newcomer node} (abbrv. newcomer). The parameter $\alpha$ denotes the number of field elements per storage node, and the parameter~$\beta$ denotes the number of field elements which are to be downloaded from each helper node during repair. Further requirements are the ability to reconstruct from \textit{any} set of~$k$ nodes, and to repair from \textit{any} set of~$d$ nodes. 

 In~\cite{CutSetBound}, the parameters of any regenerating code were shown to satisfy the so called \textit{cut-set} bound

\begin{equation}\label{eqn:cutset}
B\le \sum_{i=0}^{k-1}\min \{\alpha,(d-i)\beta\},
\end{equation}

from which a tradeoff between $\alpha$ and $\beta$ is apparent. One point of this tradeoff, in which $\alpha$ is minimized, attains $\alpha=\frac{B}{k}$, whereas the second point, in which $\beta$ is minimized, attains $\alpha=\beta d$. Codes which attain~\eqref{eqn:cutset} with equality and have $\alpha=\frac{B}{k}$ are called Minimum Storage Regenerating (MSR) codes. Codes which attain~\eqref{eqn:cutset} with equality and have $\alpha=d\beta$ are called Minimum Bandwidth Regenerating (MBR) codes.

In the first part of this paper, regenerating codes with $\alpha=d\beta$ are constructed. These codes have~$B$ that asymptotically attains~\eqref{eqn:cutset} with equality as~$k$ increases, and is close to attaining equality even for small values of~$k$. In addition, as long as the file size is large enough, these codes may be realized over any given field, and in particular, the binary field. This restriction on the file size is usually satisfied in typical distributed storage systems. The second part of the paper contains a construction of regenerating codes with $d\ge 2k-2$ that have~$B$ which approaches~$\alpha k$ as~$k$ increases. As in the first part, these codes may also be realized over any given field if the file size is large enough. 

Conceptually, this construction serves both as a mathematical proof of concept that almost optimal regenerating codes exist over any field, and as a means to reduce the complexity of the involved encoding, reconstruction, and repair algorithms by using smaller finite field arithmetics. Note that by employing the structure of an extension field as a vector space over its base field, one may implement any code over an extension field by operations over the base field. However, this approach requires implementation of sophisticated circuits for multiplication over the extension field, while employing the base field itself enables implementation using matrix multiplication only.

\ifdefined\subspaces
The last part presents an alternative approach for the construction in the first part. This approach employs newly studied algebraic objects called \textit{subspace codes}, which recently gained increasing attention due to their application in error correction for random network coding~\cite{KK}. This approach does not result in an improvement of the first part of the paper, and yet it presents applications of mathematical objects of independent interest, and a potential for future improvements of the results.
\fi

Our techniques are inspired by the code construction in~\cite{PMcodes}, which attains MBR codes for all parameters $n,k$, and~$d$, as well as MSR codes for all parameters $n,k$, and~$d$ such that $d\ge 2k-2$, where both constructions have $\beta=1$. It is noted in~\cite[Sec.~I.C]{PMcodes} that only the case $\beta=1$ is discussed since \textit{striping of data} is possible, and larger $\beta$ may be obtained by code concatenation. It will be shown in the sequel that allowing a larger $\beta$, \textit{not} through concatenation, enables a significant reduction in field size with a small and often negligible loss of code rate.

According to~\cite[Sec.~I.B.]{PMcodes}, regenerating codes which do not attain~\eqref{eqn:cutset} with equality are not MBR codes, even if they satisfy $\alpha=d\beta$ and attain~\eqref{eqn:cutset} asymptotically. Similarly, regenerating codes which attain $B=\alpha k$ asymptotically are not considered MSR codes. To the best of our knowledge, such codes were not previously studied, and hence, we coin the following terms.

\begin{definition}\label{definition:NMBR-NMSR}
	A regenerating code is called a \emph{nearly MBR} (NMBR) code if it satisfies~$\alpha=\beta d$, and~$B$ approaches the cut-set bound~\eqref{eqn:cutset} as~$k$ increases. Similarly, a regenerating code is called a \emph{nearly MSR} (NMSR) code if~$B$ approaches~$\alpha k$ as~$k$ increases.
\end{definition}

This paper is organized as follows. Previous work is discussed in Section~\ref{section:PreviousWork}. Mathematical background on several notions from field theory, number theory, and matrix analysis is given in Section~\ref{section:Preliminaries}. NMBR codes are given in Section~\ref{section:MBR}, and NMSR codes in Section~\ref{section:MSR}, each of which contains a subsection with a detailed asymptotic analysis and numerical examples. 
\ifdefined\subspaces
The aforementioned connection to subspace codes is given in Section~\ref{section:SubspaceCodes}, and concluding
remarks with future research directions are given in Section~\ref{section:Discussion}. Finally, omitted proofs and additional properties of the given codes are given in the Appendices.
\else
Finally, concluding remarks with future research directions are given in Section~\ref{section:Discussion}.
\fi

\section{Previous Work}\label{section:PreviousWork}
MBR codes for all parameters $n,k$, and~$d$ were constructed in~\cite{PMcodes}, where the underlying field size $q$ must be at least $n$, and $B={k+1 \choose 2}+k(d-k)$. MSR codes for all parameters $n,k$, and~$d$ such that $d\ge 2k-2$ were also constructed in~\cite{PMcodes}, where the underlying field size must be at least~$n(d-k+1)$, and $B=(d-k+1)(d-k+2)$. These codes are given under a powerful framework called \textit{product matrix codes}, and are the main objects of comparison in this paper. Henceforth, these codes are denoted by PM-MBR and PM-MSR, respectively.

Broadly speaking, the construction of PM-MBR codes associates a distinct field element~$\gamma$ to each storage node, which stores $(1,\gamma,\gamma^2,\ldots,\gamma^d)\cdot M$, where $M$ is a symmetric matrix that contains the file~$x$. In our NMBR construction we replace the vector $(1,\gamma,\gamma^2,\ldots,\gamma^d)$ by a properly chosen matrix. This matrix is associated with an element of an \textit{extension field} of the field~$\bF_q$, an approach which enables a reduction in field size. Our NMSR construction uses similar notions, where the proofs are a bit more involved, and require tools from basic number theory and matrix analysis.

Product matrix codes were recently improved in~\cite{BASICcodes}. The improvement is obtained by operating over a ring~$\cR_m$ in which addition and multiplication may be implemented by cyclic shifts and binary additions. The size of $\cR_m$ is $2^m$, where $m$ must not be divisible by $2,\ldots,n-1$~\cite[Th.~10]{BASICcodes}. Using our techniques, it is possible to employ the binary field itself (or any other given field), rather than the aforementioned ring~$\cR_m$, with minor loss of rate. PM-MSR codes were also recently improved in~\cite{MSRq>n}, which reduced the required field size to $q>n$, whenever $q$ is a power of two. This result follows from a special case of our construction (see Remark~\ref{remark:MSRk=b} in Section~\ref{section:MSR}).

A closely related family of MBR codes called \textit{repair-by-transfer} codes is discussed in~\cite{ReplicationMBR}. In repair-by-transfer codes a node which participates in a repair must transfer its data without any additional computations. In~\cite{ReplicationMBR}, the field size required for repair-by-transfer MBR codes is reduced to $O(n)$ instead of $O(n^2)$ in previous constructions~\cite{RBTcodes}. A similar result is obtained by~\cite{NovelRepair}, which also studied the repair and reconstruction complexities. In addition,~\cite{ReplicationMBR} obtain \textit{binary} repair-by-transfer codes for the special cases $k=d=n-2$ and $k+1=d=n-2$. 

Further aspects of regenerating codes where thoroughly studied in recent years~\cite{collective1,collective2,collective3,collective4,collective5,collective6,Us1}. In particular, the problem of constructing high rate MSR codes, i.e., with a constant number of parity nodes, has received a great deal of attention~\cite{SashaAndYe,SashaAndYe2,HighrateMSR1,HighrateMSR2,Us2}. Implementing our techniques for high rate MSR codes is one of our future research directions.

\section{Preliminaries}\label{section:Preliminaries}
This section lists several notions from field theory, linear algebra, number theory and matrix analysis, which are required for the constructions that follow. To this end, the following notations are introduced. For an integer~$m$, the notation~$\bZ_m$ stands for the ring of integers modulo~$m$, and~$[m]\triangleq\{1,\ldots,m\}$. The ring of univariate polynomials over~$\bF_q$ is denoted by~$\bF_q[x]$. For integers~$s$ and~$t$, the notations~$\bF_{q}^{s\times t}$ stands for the ring of~$s\times t$ matrices over~$\bF_q$. For a matrix $A$, let $A_{i,j}$ be its $(i,j)$-th entry, and let~$A_i$ be its~$i$-th row or column, where ambiguity is resolved if unclear from context. If~$A$ is a matrix in $\bF_q^{ms\times mt}$ which consists of~$s\cdot t$ blocks of size $m\times m$ each, we denote its $(i,j)$-th block by~$\llbracket A\rrbracket ^{(m)}_{i,j}$, and omit the notation~$(m)$ if it is clear from the context. The notations $I_m$ and $\bold{0}_m$ are used to denote the identity and zero matrix of order~$m$, respectively. 

\subsection{Companion matrices and representation of extension fields}
\begin{definition}\label{definition:CompanionMatrix}
	The companion matrix of a monic univariate polynomial $P(x)=p_0+p_1x+\ldots+p_{e-1}x^{e-1}+x^e\in \bF_q[x]$ is the~$e\times e$ matrix
	\begin{align*}
	\begin{pmatrix}
	0 & 0 & \cdots & -p_0 \\
	1 & 0 & \cdots & -p_1 \\
	0 & \ddots & \ddots & \vdots \\
	0 & \cdots & 1& -p_{e-1} 
	\end{pmatrix}.
	\end{align*}
\end{definition}

It is an easy exercise to show that the minimal and characteristic polynomials of a companion matrix are its corresponding polynomial, and the eigenvalues are the roots of that polynomial (which may reside in an extension field of the field of coefficients).

The following lemma, which is well-known, provides a convenient yet redundant representation of extension fields as matrices over the base field. Unlike other representations, this representation encapsulates both the additive and the multiplicative operations in the extension field, both as the respective operations between matrices.

\begin{lemma}\cite[Ch.~2, Sec.~5]{LidlAndNiederreiter}\label{lemma:ExtensionFieldRepresentation}
	If $P\in\bF_{q}[x]$ is monic and irreducible of degree~$m$ with companion matrix~$M_P$, then the linear span over~$\bF_q$ of the set~$\{M_P^i\}_{i=0}^{m-1}$ is isomorphic to~$\bF_{q^m}$. If~$P$ is also primitive, then~$\{M_P^i\}_{i=0}^{q^m-2}\cup\{0\}$ is isomorphic to~$\bF_{q^m}$.
\end{lemma}

Lemma~\ref{lemma:ExtensionFieldRepresentation} also has an inverse~\cite{Wardlaw}. That is, given the field $\bF_{q^m}$, it is possible to represent its elements as all powers of the companion matrix~$P$ which corresponds to an irreducible polynomial of degree~$m$ over~$\bF_q$.
Hence, for any~$m$ and any such matrix~$P$, let~$\theta_P:\bF_{q^m}\to\bF_q^{m\times m}$ be the function which maps an element in the extension field~$\bF_{q^m}$ to its matrix representation in~$\bF_{q}^{m\times m}$ as a linear combination of powers of~$P$, and since our results are oblivious to the choice of~$P$, we denote~$\theta_P$ by~$\theta$. Notice that~$\theta$ is a \textit{field isomorphism}, that is, every~$y_1$ and~$y_2$ in~$\bF_{q^m}$ satisfy that $\theta(y_1\cdot y_2)=\theta(y_1)\cdot\theta(y_2)$ and $\theta(y_1+y_2)=\theta(y_1)+\theta(y_2)$. The function~$\theta$ can be naturally extended to matrices, where~$A\in\bF_{q^m}^{s\times t}$ is mapped to
\begin{align}\label{eqn:Theta}
\Theta(A)\triangleq 
\begin{pmatrix}
\theta(A_{1,1}) & \theta(A_{1,2}) & \cdots & \theta(A_{1,t})\\
\theta(A_{2,1}) & \theta(A_{2,2}) & \cdots & \theta(A_{2,t})\\
\vdots &\vdots &\ddots &\vdots \\
\theta(A_{s,1}) & \theta(A_{s,2}) & \cdots & \theta(A_{s,t})\\
\end{pmatrix}\in\bF_q^{ms\times mt}.
\end{align}

\begin{lemma}\label{lemma:ThetaMult}
	For any integers $m,s,t$ and~$\ell$, if $A\in\bF_{q^m}^{s\times t}$ and $B\in\bF_{q^m}^{t\times \ell}$ then $\Theta(AB)=\Theta(A)\cdot \Theta(B)$.
\end{lemma}

\begin{proof}
	By the definition of~$\Theta$, and by using the fact that~$\theta$ is a field isomorphism, for all $i\in[s]$ and $j\in[\ell]$ we have that
	\begin{align*}
	\llbracket \Theta(AB) \rrbracket^{(m)}_{i,j}&=\theta\left(\sum_{k=1}^{t}A_{i,k}B_{k,j}\right)=\sum_{k=1}^{t}\theta\left(A_{i,k}\right)\theta\left(B_{k,j}\right)\\
	&=\sum_{k=1}^{t}\llbracket \Theta(A)\rrbracket_{i,k}^{(m)}\llbracket \Theta(B)\rrbracket_{k,j}^{(m)}=\llbracket \Theta(A)\cdot \Theta(B) \rrbracket_{i,j}^{(m)}.
	\end{align*}
\end{proof}

\begin{lemma}\label{lemma:ThetaInv}
	For any integers~$m$ and~$t$, if $A\in\bF_{q^m}^{t\times t}$ is invertible then~$\Theta(A)\in\bF_{q}^{mt\times mt}$ is invertible.
\end{lemma}

\begin{proof}
	According to Lemma~\ref{lemma:ThetaMult}, since $A^{-1}$ exists it follows that
	\begin{align*}
	I_{mt}=\Theta(I_t)=\Theta(A\cdot A^{-1})=\Theta(A)\cdot \Theta(A^{-1}),
	\end{align*}
	and hence~$\Theta(A^{-1})$ is the inverse of~$\Theta(A)$. 
\end{proof}

\subsection{Kronecker products and cyclotomic cosets}
The proofs of the construction of NMSR codes in Subsection~\ref{section:MSRconstruction} are slightly more involved than those given in other sections. The main tools in those proofs are cyclotomic cosets and Kronecker products, which are discussed in this subsection.

\begin{definition}\cite[Sec.~3.7]{HuffmanAndPless}, \cite[Sec.~7.5]{Ronny'sBook}\label{definition:cyclotomic}
	For an integer~$m$, a prime power~$q$ such that~$\gcd(q,m)=1$, and $s\in\bZ_m$, a subset of~$\bZ_m$ of the form~$\{s,sq,sq^2,sq^3,\ldots\}$ is called a $q$-cyclotomic coset modulo~$m$.
\end{definition}

It is well known (e.g.,~\cite{Ronny'sBook,HuffmanAndPless}) that for any~$m$ such that~$\gcd(q,m)=1$, the size of any $q$-cyclotomic coset modulo~$m$ divides the \textit{order} of~$q$ in~$\bZ_m$ (that is, the smallest integer~$t$ such that $q^t=1~(\bmod m)$). 

\begin{definition}\label{definition:Kronecker}
	For a matrix~$A\in\bF_q^{s\times t}$ and a matrix $B\in\bF_q^{n\times m}$, the Kronecker product~$A\otimes B$ is the matrix
	\begin{align*}
	\begin{pmatrix}
	A_{1,1}B & A_{1,2}B & \cdots & A_{1,t}B \\
	A_{2,1}B & A_{2,2}B & \cdots & A_{2,t}B \\
	\vdots   & \vdots   & \ddots & \vdots \\
	A_{s,1}B & A_{s,2}B & \cdots & A_{s,t}B  
	\end{pmatrix}\in\bF_{q}^{sn\times tm}.
	\end{align*}
\end{definition}

The Kronecker product is useful when solving equations in which the unknown variable is a matrix. This application is enabled through an operator called~$\vect$, defined as follows.

\begin{definition}\cite[Def.~1]{Neudecker}
	For a matrix~$A\in\bF_q^{s\times t}$ with columns $A_1,\ldots,A_t$, let
	\begin{align*}
	\vect(A)\triangleq 
	\begin{pmatrix}
	A_{1}\\
	A_{2}\\
	\vdots \\
	A_{t}
	\end{pmatrix}\in \bF_q^{st}.
	\end{align*}
\end{definition}

The following two lemmas present several properties of the Kronecker product and the~$\vect$ operator. These lemmas are well known, and their respective proofs may be found, e.g., in~\cite{Laub,HowAndWhy,Neudecker}. In particular, Lemma~\ref{lemma:Vectorizing} which follows discusses a close variant of the so called \textit{Sylvester equation}~$AX+XB=C$, where $A,B,$ and~$C$ are known matrices, and~$X$ is an unknown matrix. For completeness, full proofs are detailed below.

\begin{lemma}\label{lemma:Vectorizing}
	For an integer~$m$, if $A,X$, and $B$ are $m\times m$ matrices over~$\bF_q$, then $\vect(AXB-X)=\left(B^\top\otimes A-I_{m^2}\right)\cdot\vect(X)$.
\end{lemma}

\begin{proof}
	Clearly, if $X_1,\ldots,X_m$ are the columns of~$X$ and $B_1,\ldots,B_m$ are the columns of~$B$, then the $i$-th column of $(AXB-X)$ is 
	\begin{align*}
	AXB_i-X_i&=\sum_{j=1}^{m}B_{j,i}(AX)_j-X_i\\
	&= \sum_{j=1}^{m}B_{j,i}AX_j-X_i\\
	&=
	\begin{pmatrix}
	B_{1,i}A & B_{2,i}A & \cdots & B_{i-1,i}A & B_{i,i} A-I_m& B_{i+1,i}A & \ldots & B_{m,i}A
	\end{pmatrix} \cdot
	\begin{pmatrix}
	X_1\\X_2\\ \vdots \\ X_m
	\end{pmatrix},
	\end{align*}
	and hence, according to Definition~\ref{definition:Kronecker}, it follows that
	\begin{align*}
	\vect (AXB-X)&=
	\begin{pmatrix}
	B_{1,1}A-I_m & B_{2,1}A & \ldots & B_{m,1}A\\
	B_{1,2}A & B_{2,2}A-I_m & \ldots & B_{m,2}A\\
	\vdots & \vdots & \ddots &  \vdots\\
	B_{1,m}A & B_{2,m}A &  \ldots & B_{m,m}A-I_m\\ 
	\end{pmatrix}\cdot
	\begin{pmatrix}
	X_1\\X_2\\ \vdots \\ X_m
	\end{pmatrix}\\
	&= \left(B^\top \otimes A -I_{m^2}\right)\cdot \vect(X).
	\end{align*}
\end{proof}

\begin{lemma}\label{lemma:KroneckerEigenvalues}
	If~$A$ and~$B$ are two~$m\times m$ matrices over~$\bF_q$, and~$\bF_{q^\ell}$ is a field which contains all eigenvalues $\lambda_1,\ldots,\lambda_m$ of~$A$ and $\mu_1,\ldots,\mu_m$ of~$B$, then the eigenvalues of~$A\otimes B$ are~$\{\lambda_i\mu_j\vert i,j\in[m]\}$.
\end{lemma}

\begin{proof}
	For any~$i$ and~$j$ in~$[m]$, let~$v_i$ and~$u_j$ be (column) eigenvectors in~$\bF_{q^\ell}^m$ such that~$Av_i=\lambda_iv_i$ and~$Bu_j=\mu_ju_j$. By Definition~\ref{definition:Kronecker}, it follows that
	\begin{align*}
	\left(A\otimes B\right)\cdot \left(u_j\otimes v_i\right)&
	=\begin{pmatrix}
	B_{1,1}A & B_{1,2}A & \ldots & B_{1,m}A\\
	B_{2,1}A & B_{2,2}A & \ldots & B_{2,m}A\\
	\vdots & \vdots & \ddots &  \vdots\\
	B_{m,1}A & B_{m,2}A & \ldots & B_{m,m}A\\ 
	\end{pmatrix}\cdot 
	\begin{pmatrix}
	u_{j,1}v_i\\
	\vdots \\
	u_{j,n}v_i
	\end{pmatrix}\\
	&=
	\begin{pmatrix}
	u_{j,1}B_{1,1}Av_i + u_{j,2}B_{1,2}Av_i + \ldots + u_{j,n}B_{1,m}Av_i\\
	u_{j,1}B_{2,1}Av_i + u_{j,2}B_{2,2}Av_i + \ldots + u_{j,n}B_{2,m}Av_i\\
	\vdots \\
	u_{j,1}B_{m,1}Av_i + u_{j,2}B_{m,2}Av_i + \ldots + u_{j,n}B_{m,m}Av_i\\ 
	\end{pmatrix}\\
	&=
	\lambda_i\cdot\begin{pmatrix}
	B_{1,1}u_{j,1}v_i + B_{1,2}u_{j,2}v_i + \ldots + B_{1,m}u_{j,n}v_i\\
	B_{2,1}u_{j,1}v_i + B_{2,2}u_{j,2}v_i + \ldots + B_{2,m}u_{j,n}v_i\\
	\vdots \\
	B_{m,1}u_{j,1}v_i + B_{m,2}u_{j,2}v_i + \ldots + B_{m,m}u_{j,n}v_i\\ 
	\end{pmatrix}\\
	&=\lambda_i\cdot\begin{pmatrix}
	(Bu_j)_1v_i \\
	(Bu_j)_2v_i \\
	\vdots\\
	(Bu_j)_mv_i \\
	\end{pmatrix}
	=\lambda_i\left(Bu_j\right)\otimes(v_i)=\lambda_i\mu_j\left(u_j\otimes v_i\right).
	\end{align*}
\end{proof}

The following technical lemma will be required in the application of Lemma~\ref{lemma:Vectorizing}. Although it follows immediately from one of the common equivalent definitions of eigenvalues, a full proof is given.

\begin{lemma}\label{lemma:eigenvalue1}
	If $A$ is an ${m\times m}$ matrix over~$\bF_q$, then $A-I$ is invertible if and only if~$1$ is not an eigenvalue of~$A$.
\end{lemma}

\begin{proof}
	Assume~$A-I$ is invertible. If~$1$ is an eigenvalue of~$A$ then there exists a nonzero vector~$v\in\bF_q^m$ such that $Av=v$, and hence, $(A-I)v=v-v=0$, and hence $\ker(A-I)\ne\{0\}$, which implies that $A-I$ is \textit{not} invertible, a contradiction.
	
	Conversely, assume that~$1$ is not an eigenvalue of $A$. If~$A-I$ is not invertible then there exists a nonzero vector~$v\in\bF_q^m$ such that $(A-I)v=0$, which implies that $Av=v$, and hence~$1$ is an eigenvalue of~$A$, a contradiction.  
\end{proof}

\section{Nearly MBR codes}\label{section:MBR}
For any given $n,k,d,q$, and a sufficiently large file size~$B$, this section presents regenerating codes with $\alpha=d\beta$, and~$B$ which approaches the cut-set bound as~$k$ increases. For any such $n,k,d$ and $q$ let $b$ be an integer such that 

\begin{itemize}
	\item [A1.\namedlabel{itm:A1}{A1}] $b\ge k\log_q n$,
	\item [A2.\namedlabel{itm:A2}{A2}] $k\mid b$,
\end{itemize}
and let $B\triangleq \frac{b(b+1)}{2}+b^2\left(\frac{d}{k}-1\right)$. Notice that Condition~\ref{itm:A1} implies that $B=\Omega \left(kd\log_q(n)^2\right)$. Since usually, the file size $B$ is in the order of magnitude of billions, and the number of nodes is in the order of magnitude of dozens, Condition~\ref{itm:A1} is trivially satisfied in many distributed storage systems (see Subsection~\ref{section:MBRasymptotic} for explicit examples). 

\subsection{Construction}\label{section:MBRconstruction}
Given a file $x\in\bF_q^B$, define the following data matrix, which resembles the corresponding one in~\cite{PMcodes}:

\begin{align}\label{eqn:DataMatrix}
\nonumber X &= 
\begin{pmatrix}
S & T\\
T^\top & 0
\end{pmatrix}\in\bF_q^{\frac{db}{k}\times\frac{db}{k}}\mbox{, where}\\
S&\triangleq\begin{pmatrix}
x_1 & x_2 & x_3 & \ldots & x_b \\
x_2 & x_{b+1} & x_{b+2} & \ldots & x_{2b-1}  \\
\vdots & \vdots & \vdots & \vdots & \vdots\\
x_b & ~ &\cdots&~&x_{\frac{b(b+1)}{2}}
\end{pmatrix}\in\bF_q^{b\times b},
\end{align}
and $T\in\bF_q^{b\times b(d/k-1)}$ contains the remaining $b^2\left(\frac{d}{k}-1\right)$ elements of~$x$ in some arbitrary order. 

Let $P$ be a companion matrix of any primitive polynomial of degree $\frac{b}{k}$ over~$\bF_q$, and let $i_1,\ldots,i_n$ be distinct integers in the range $\{0,\ldots,q^{b/k}-1\}$, which exist by~\ref{itm:A1}. Using $P$ and $i_1,\ldots,i_n$, define the following encoding matrix,

\begin{align}\label{eqn:EncodingMatrix}
\nonumber M = 
\begin{pmatrix}
M_1\\
M_2 \\
\vdots\\
M_n
\end{pmatrix}\in \bF_q^{\frac{nb}{k}\times \frac{db}{k}}\mbox{, where}\\
M_j \triangleq 
\begin{pmatrix}
I & P^{i_j} & P^{2i_j} & \cdots & P^{(d-1)i_j} 
\end{pmatrix}\in\bF_q^{\frac{b}{k}\times \frac{db}{k}},
\end{align}
and store $M_j\cdot X$ in storage node $j$. Notice that by the definition of the matrix~$P$, we have that~${\alpha=\frac{b^2}{k^2}\cdot d}$.

\begin{remark}
	It is possible to replace~$P$ by the companion matrix of an irreducible polynomial which is not necessarily primitive, in which case, let~$\{A_j(P)\}_{j=1}^n$ be distinct nonzero linear combinations of $\{P^i\}_{i=0}^{b/k-1}$, and define $	M_j\triangleq 
	\begin{pmatrix}
	I & A_j(P) & A_j(P)^2 & \cdots & A_j(P)^{d-1} 
	\end{pmatrix}$. However, we choose a primitive polynomial for convenience.
\end{remark}

\begin{remark}\label{remark:MBRk=b}
	For $k=b$ this code is a PM-MBR code~\cite[Sec.~IV]{PMcodes}, and in which case Condition~\ref{itm:A1} implies that $q\ge n$. Therefore, the advantage of our techniques exists only for $b>k$.
\end{remark}

\begin{theorem}\label{theorem:MBRrepair}
	In the above code, exact repair of any failed node may be achieved by downloading $\beta\triangleq\frac{b^2}{k^2}$ field elements from any $d$ of the remaining nodes.
\end{theorem}

\begin{proof}
	Assume that node $i$ failed, and $D=\{j_1,\ldots,j_d\}$ is a subset of $[n]$ of size $d$ such that $i\notin D$. To repair node $i$, every node $j_t\in D$ computes $M_{j_t}XM_i^\top$, which is a $\frac{b}{k}\times\frac{b}{k}$ matrix over~$\bF_q$, and sends it to the newcomer. The newcomer obtains
	\begin{align*}
	\begin{pmatrix}
	M_{j_1}XM_i^\top\\M_{j_2}XM_i^\top\\\vdots\\M_{j_d}XM_i^\top
	\end{pmatrix}=
	\begin{pmatrix}
	M_{j_1}\\M_{j_2}\\\vdots\\M_{j_d}
	\end{pmatrix}\cdot X \cdot M_i^\top\triangleq M_DXM_i^\top.
	\end{align*}
	
	According to~\eqref{eqn:EncodingMatrix}, the matrix~$M_D$ is of the form
	\begin{align*}
	M_D = 
	\begin{pmatrix}
	I & P^{i_{j_1}} & P^{2i_{j_1}} & \ldots & P^{(d-1)i_{j_1}}\\
	I & P^{i_{j_2}} & P^{2i_{j_2}} & \ldots & P^{(d-1)i_{j_2}}\\
	\vdots & 	\vdots & 	\vdots & 	\vdots & 	\vdots \\
	I & P^{i_{j_d}} & P^{2i_{j_d}} & \ldots & P^{(d-1)i_{j_d}}
	\end{pmatrix}.
	\end{align*}
	Since $M_D$ can be written as~$\Theta(M_D')$ for some invertible Vandermonde matrix~$M_D'$ in~$\bF_{q^{b/k}}^{d\times d}$, it follows by Lemma~\ref{lemma:ThetaInv} that~$M_D$ is invertible. Thus, the newcomer may multiply from the left by $M_D^{-1}$ and obtain $XM_i^\top$. Since $X$ is a symmetric matrix, exact repair is obtained by transposing.
\end{proof}

\begin{theorem}\label{theorem:MBRreconstruction}
	In the above code, reconstruction may be achieved by downloading $\alpha=\frac{b^2}{k^2}\cdot d$ field elements per node from any $k$ nodes.
\end{theorem}

\begin{proof}
	Let $K=\{j_1,\ldots,j_k\}$ be a subset of $[n]$ of size $k$, and download $M_{j_i}X$ from node $j_i$ for each $j_i\in K$. The data collector thus obtains 
	\begin{align*}
	M_K \cdot X&\triangleq 
	\begin{pmatrix}
	I & P^{i_{j_1}} & P^{2i_{j_1}} & \ldots & P^{(d-1)i_{j_1}}\\
	I & P^{i_{j_2}} & P^{2i_{j_2}} & \ldots & P^{(d-1)i_{j_2}}\\
	\vdots & 	\vdots & 	\vdots & 	\vdots & 	\vdots \\
	I & P^{i_{j_k}} & P^{2i_{j_k}} & \ldots & P^{(d-1)i_{j_k}}
	\end{pmatrix}\cdot 
	\begin{pmatrix}
	S & T \\
	T^\top & 0
	\end{pmatrix} \\
	&\triangleq 
	\begin{pmatrix}
	M_{K}'S+M_{K}''T^\top & M_{K}'T
	\end{pmatrix},
	\end{align*}
	where
	\begin{align*}
	M_{K}' &\triangleq 
	\begin{pmatrix}
	I & P^{i_{j_1}} & P^{2i_{j_1}} & \ldots & P^{(k-1)i_{j_1}}\\
	I & P^{i_{j_2}} & P^{2i_{j_2}} & \ldots & P^{(k-1)i_{j_2}}\\
	\vdots & 	\vdots & 	\vdots & 	\vdots & 	\vdots \\
	I & P^{i_{j_k}} & P^{2i_{j_k}} & \ldots & P^{(k-1)i_{j_k}}
	\end{pmatrix}\mbox{, and} & 
	M_{K}'' &\triangleq 
	\begin{pmatrix}
	P^{ki_{j_1}} & P^{2i_{j_1}} & \ldots & P^{(d-1)i_{j_1}}\\
	P^{ki_{j_2}} & P^{2i_{j_2}} & \ldots & P^{(d-1)i_{j_2}}\\
	\vdots & 	\vdots & 	\vdots & 	\vdots \\
	P^{ki_{j_k}} & P^{2i_{j_k}} & \ldots & P^{(d-1)i_{j_k}}
	\end{pmatrix}.
	\end{align*}
	As in the proof of Theorem~\ref{theorem:MBRrepair}, we have that $M_K'$ is invertible. Hence, if $d>k$, the matrix $T$ may be restored by extracting the~$b^2\left(\frac{d}{k}-1\right)$ rightmost columns of $M_KX$ and multiplying by $(M_K')^{-1}$. Having~$T$, it can be used to reduce $M_K''T^\top$ from the remaining columns of~$M_KX$, and then extracting~$S$ is similar. If $d=k$, then $X=S$, and multiplication by $M_K=M_K'$ suffices for reconstruction.
\end{proof}

By Theorem~\ref{theorem:MBRrepair} and Theorem~\ref{theorem:MBRreconstruction} it is evident that $\alpha=d\beta$, and hence this construction attains minimum bandwidth repair. In Subsection~\ref{section:MBRasymptotic} it will be shown that although the cut-set bound is not attained with equality,~$B$ approaches the cut-set bound~\eqref{eqn:cutset} as~$k$ increases. Moreover, it will be evident that a small and often negligible loss of rate is obtained already for small values of~$k$.

\subsection{The proximity of NMBR codes to MBR codes}\label{section:MBRasymptotic}
In this subsection it is shown that the codes constructed in Subsection~\ref{section:MBRconstruction} do not attain the cut-set bound~\eqref{eqn:cutset}, and hence cannot be considered MBR codes even though they attain $\alpha=d\beta$ (see Definition~\ref{definition:NMBR-NMSR}, and its preceding discussion). However, it is also shown that the cut-set bound is nearly achieved for large enough $k$, together with few specific examples which demonstrate a small loss of rate.

Let $C\triangleq \sum_{i=0}^{k-1}\min\{\alpha,(d-i)\beta\}$, and recall that by~\eqref{eqn:cutset} we have that $B\le C$ for all regenerating codes. Clearly, for codes which attain $\alpha=d\beta$ we have that
\begin{align}
\nonumber C &= \sum_{i=0}^{k-1}\min\{d\beta,(d-i)\beta\}\\
~ &= \beta\sum_{i=0}^{k-1}(d-i)=\beta \left(dk-\frac{k(k-1)}{2}\right).\label{eqn:cutsetMBR}
\end{align}

Hence, for the codes which are presented in Subsection~\ref{section:MBRconstruction} we have that $C=\frac{b^2}{k}\left(d-\frac{k-1}{2}\right)$. It is readily verified that indeed, $C>B$, thus~\eqref{eqn:cutset} is not attained, and hence these are not MBR codes. However, we have that 
\begin{align*}
\frac{B}{C}=\frac{2d-k\left(2-\frac{b+1}{b}\right)}{2d-k+1}=\frac{2-\frac{k}{d}\cdot\frac{b-1}{b}}{2-\frac{k}{d}+\frac{1}{d}}
\end{align*}
and hence the cut-set bound is achieved in the asymptotic regime. That is, since a large $k$ implies a large $b$ (since $k\vert b$ in Condition~\ref{itm:A2}) and a large~$d$ (since $d\ge k$), by following the outline of Subsection~\ref{section:MBRconstruction} and choosing a large enough $k$, one may obtain a code in which $B$ is arbitrarily close to~$C$, regardless of the relation between~$k$ and~$d$. 

In the remainder of this section, a detailed comparison of parameters between the PM-MBR codes and our NMBR codes is given. From these examples it will be evident that the decrease in file size (in comparison with the cut-set bound), and hence the decrease in the code rate, is a small price to pay for a considerable reduction in field size.

The curious reader might suggest that the extension field representation which is given in Lemma~\ref{lemma:ExtensionFieldRepresentation}, can be applied directly to PM-MBR codes over an extension field, obtaining regenerating codes over the respective base field. This intuition is formalized in the following definition. For this definition, recall that PM-MBR codes may be obtained by choosing~$k=b$ in the construction in Subsection~\ref{section:MBRconstruction}.

\begin{definition}\label{definition:EPM-MBR}
	Given a PM-MBR code over an extension field~$\bF_{q^m}$ with an encoding matrix~$M$ and data matrix~$X$, let EPM-MBR be the code over~$\bF_q$ which results from applying the function~$\Theta$ from~\eqref{eqn:Theta} on the encoding matrix~$M$ and multiplying it by a data matrix~$X'$. The data matrix~$X'$ is given by applying~$\theta$ on the upper triangular part of the data matrix~$X$, and completing 
	the lower triangular part to obtain symmetry.
\end{definition}

In order to apply the repair and reconstruction algorithms from Theorem~\ref{theorem:MBRrepair} and Theorem~\ref{theorem:MBRreconstruction} to EPM-MBR codes, the data matrix~$X'$ must be symmetric. Hence, it follows that the only data matrices~$X\in \bF_{q^m}^{d\times d}$ on which EPM-MBR codes maintain their repair and reconstruction capabilities are those in which all diagonal submatrices $\theta(X_{i,i})$ of~$X'$ are symmetric. Since companion matrices are in general not symmetric, this usually induces a further loss of rate. For simplicity, we shall ignore this detail in the comparison which follows, since NMBR codes will be shown to supersede EPM-MBR codes even without this additional rate loss. In the remainder of this section we compare between EPM-MBR codes, NMBR codes, and PM-MBR codes with the concise vector space representation of extension field elements\footnote{I.e., each extension field element is represented by a vector over the base field.}.

In PM-MBR codes the file size $B$ is a function of $k$ and $d$, and in addition, $\beta=1$. Further, all parameters are measured in field elements rather than in bits. Therefore, to achieve a fair comparison, one must concatenate a PM-MBR code to itself in order to obtain the same parameters $n,k,d,\alpha,$ and~$\beta$ when measured in bits, and only then compare the resulting $B$, $q$, and the rate $\frac{B}{\alpha n}$. In addition, since fields of even characteristic are essential for hardware implementation, we restrict our attention to $q=2$ in our codes, and to $q$ which is an integer power of $2$ for PM-MBR codes. Hence, the PM-MBR code is concatenated with itself $\frac{b^2}{\ceil{\log n}k^2}$ times, and considered with $q=2^{\ceil{\log n}}$ (the smallest integer power of two that is at least~$n$), where each element in this field is represented by a vector in $\bF_2^{\ceil{\log n}}$. Similarly, the EPM-MBR code is concatenated with itself $\frac{b^2}{\ceil{\log n}^2k^2}$ times, and considered with the same $q=2^{\ceil{\log n}}$, where each element in this field is represented by a square matrix in $\bF_2^{\ceil{\log n}\times \ceil{\log n}}$.

Notice that MBR codes have $B=\beta \left(dk-k(k-1)/2\right)$ (see~\eqref{eqn:cutsetMBR}), where $B$ is measured in elements over $\bF_q$. Therefore, by setting $\beta=1$, $q=2^{\ceil{\log n}}$, and concatenating a PM-MBR code $\frac{b^2}{\ceil{\log n}k^2}$ times with itself, we have that the number of information bits in the file is $C=\frac{b^2}{k}\left(d-\frac{k-1}{2}\right)$. Similarly, by concatenating an EPM-MBR code $\frac{b^2}{\ceil{\log n}^2k^2}$ times with itself, since each field element is represented by a $\ceil{\log n}\times \ceil{\log n}$ binary matrix that contains~$\ceil{\log n}$ information bits, it follows that the number of information bits in the file is $\beta(dk-k(k-1)/2)\cdot \frac{b^2}{\ceil{\log n}^2k^2} \cdot \ceil{\log n}=\frac{C}{\ceil{\log n}}$. As a result, by fixing any $n,k,$ and $d$ such that $k\le d\le n-1$, we have Table~\ref{table:GeneralParameters}, in which the values of $\beta,\alpha$, and~$B$ are given in bits.

\renewcommand{\arraystretch}{1.35}
\newcommand{\specialcell}[2][c]{%
	\begin{tabular}[#1]{@{}c@{}}#2\end{tabular}}
\begin{table}[h]
	\centering
	\begin{tabular}{|c|c|c|c|}
		\cline{2-4}
		\multicolumn{1}{c}{~} &\multicolumn{1}{|c|}{NMBR} & \specialcell[c]{PM-MBR \\ concatenated $\frac{b^2}{\ceil{\log n}k^2}$ times}&\specialcell[c]{EPM-MBR\\concatenated $\frac{b^2}{\ceil{\log n}^2k^2}$ times}\\ \hhline{-===}
		$q$ & $2$ & $2^{\ceil{\log n}}$&2\\ \hline
		$\beta$ & $\frac{b^2}{k^2}$  & \specialcell[c]{$\frac{b^2}{\ceil{\log n}k^2}$ field elements \\in vector form,\\i.e., $\frac{b^2}{k^2}$ bits.} & \specialcell[c]{$\frac{b^2}{\ceil{\log n}^2k^2}$ field elements\\in matrix form,\\i.e., $\frac{b^2}{k^2}$ bits.}\\ \hline
		$\alpha$& $\frac{b^2}{k^2}\cdot d$  & $\frac{b^2}{k^2}\cdot d$ & $\frac{b^2}{k^2}\cdot d$ \\ \hline
		$B$ & $\frac{b(b+1)}{2}+b^2\cdot\left(\frac{d}{k}-1\right)$ & $\frac{b^2}{k}\left(d-\frac{k-1}{2}\right)$ & $\frac{b^2}{k\ceil{\log n}}\left(d-\frac{k-1}{2}\right)$\\ \hline
		Rate & $\frac{k^2}{dn}\cdot\left(\frac{d}{k}-\frac{1}{2}+\frac{1}{2b}\right)$ & 	$\frac{k^2}{dn}\cdot\left(\frac{d}{k}-\frac{k-1}{2k}\right)$ & $\frac{k^2}{dn\ceil{\log n}}\cdot\left(\frac{d}{k}-\frac{k-1}{2k}\right)$\\ \hline
	\end{tabular}
	\caption{A comparison of parameters between our NMBR codes (Subsection~\ref{section:MBRconstruction}) and the PM-MBR codes~\cite[Sec.~IV]{PMcodes} for general $n,k,d$.}\label{table:GeneralParameters}
\end{table}
\renewcommand{\arraystretch}{1}

Table~\ref{table:SpecificParamters} contains specific examples of the comparison given in Table~\ref{table:GeneralParameters}. The parameter~$b$ is chosen such that $\frac{b^2}{\ceil{\log n}k^2}$ and $\frac{b^2}{\ceil{\log n}^2k^2}$ are integers, and such that the resulting file size is within one of several common use cases. Notice that much smaller values of~$b$ may be chosen, for example, if one wishes to increase concurrency by code concatenation. For convenience, some values are given in either MegaBytes (MB), GigaBytes (GB), or TeraBytes (TB) rather than in bits. 

\begin{table}[h]
	\makebox[\textwidth][c]{
	\centering
	\begin{tabular}{|c|c|c|c|c|c|c|c|c|c|}
		\cline{2-10}
		\multicolumn{1}{c}{~}     & \multicolumn{1}{|c|}{$n$} & $k$ & $d$ & $\alpha$ & $\beta$ & $b$ & $q$ & $B$ & Rate \\\hhline{-=========}
		NMBR &\multirow{3}{*}{30}&\multirow{3}{*}{20}&\multirow{3}{*}{20}&\multirow{3}{*}{$250$MB}&\multirow{3}{*}{$12.5$MB}&\multirow{3}{*}{$10000\cdot k$}&2&$\approx 2.5$GB&$\approx 0.33$\\\cline{1-1}\cline{8-10}
		PM-MBR   &~					&~                 &~			      &~					 &~					   &~				  &32&$2.625$GB&$0.35$\\\cline{1-1}\cline{8-10}
		EPM-MBR   &~					&~                 &~			      &~					 &~					   &~				  &2&$0.525$GB&$0.07$\\ \hline\hline
		NMBR &\multirow{3}{*}{26}&\multirow{3}{*}{22}&\multirow{3}{*}{24}&\multirow{3}{*}{$\approx 841.7$MB}&\multirow{3}{*}{$\approx 35.07$MB}&\multirow{3}{*}{$16750\cdot k$}&2&$\approx 10.03$GB&$\approx 0.4583$\\\cline{1-1}\cline{8-10}
		PM-MBR   &~					&~                 &~			      &~					 &~					   &~				  &$32$&$\approx 10.41$GB&$\approx 0.4759$\\\cline{1-1}\cline{8-10}
		EPM-MBR   &~					&~                 &~			      &~					 &~					   &~				  &$2$&$\approx 2.8$GB&$\approx 0.095$\\\hline\hline
		NMBR &\multirow{3}{*}{260}&\multirow{3}{*}{220}&\multirow{3}{*}{240}&\multirow{3}{*}{$\approx 8.416$GB}&\multirow{3}{*}{$\approx35.06$MB}&\multirow{3}{*}{$16749\cdot k$}&2&$\approx 1.002$TB&$\approx0.4583$\\\cline{1-1}\cline{8-10}
		PM-MBR   &~					&~                 &~			      &~					 &~					   &~				  &$512$&$\approx 1.006$TB&$\approx 0.4601$\\\cline{1-1}\cline{8-10}
		EPM-MBR   &~					&~                 &~			      &~					 &~					   &~				  &$2$&$\approx 0.11$TB&$\approx 0.05$\\\hline\hline
		NMBR &\multirow{3}{*}{2600}&\multirow{3}{*}{2200}&\multirow{3}{*}{2400}&\multirow{3}{*}{$\approx84.18$GB}&\multirow{3}{*}{$\approx35.07$MB}&\multirow{3}{*}{$16752\cdot k$}&2&$\approx100.32$TB&$\approx0.4583$\\\cline{1-1}\cline{8-10}
		PM-MBR   &~					&~                 &~			      &~					 &~					   &~				  &4096&$\approx100.36$TB&$\approx0.4585$\\\cline{1-1}\cline{8-10}
		EPM-MBR   &~					&~                 &~			      &~					 &~					   &~				  &2&$\approx 8.36$TB&$\approx0.038$\\\hline
	\end{tabular}}
	\caption{A comparison of parameters between our NMBR codes (Subsection~\ref{section:MBRconstruction}) and the PM-MBR codes~\cite[Sec.~IV]{PMcodes} for several common parameters $n,k,d$.}\label{table:SpecificParamters}
\end{table}

From Table~\ref{table:SpecificParamters} it is evident that in comparison with PM-MBR codes, a considerable reduction in field size is obtained by our codes, even for rather small values of~$k$. Furthermore, our techniques obtain a larger rate in comparison with EPM-MBR codes, which are implemented over the binary field as well.

In many practical applications~\cite[Slide~38]{Plank}, multiplication in a finite field~$\bF_{2^w}$ is implemented by table look-ups for $w\le 8$, and sometimes considered infeasible in large systems with $w>8$, since it requires numerous table look-ups and expensive arithmetic. Hence, for $n>2^8=256$, our techniques improve the feasibility of storage codes without compromising the code rate significantly. 

\section{Nearly MSR codes}\label{section:MSR}
In this section, for any given $n,k,d,q$ such that $d\le n-1$ and $d=2k-2$, and for a sufficiently large file size~$B$, regenerating codes in which $B$ approaches~$\alpha k$ as~$k$ increases are provided. Codes for $d>2k-2$ with similar properties are obtained in the sequel from this construction. For any such $n,k,d$ and $q$, let $b$ be an integer such that 

\begin{itemize}
	\item [B1.\namedlabel{itm:B1}{B1}] $n\le\frac{q^{b/k}-1}{g\cdot\frac{b}{k}}$, where $g\triangleq \gcd(k-1,q^{b/k}-1)$, 
	\item [B2.\namedlabel{itm:B2}{B2}] $k\mid b$,
\end{itemize}
and let
\begin{align*}
B\triangleq \frac{b(k-1)}{k}\cdot\left(\frac{b(k-1)}{k}+1\right)=\frac{b^2(k-1)}{k}\left(1-\frac{1}{k}+\frac{1}{b}\right).
\end{align*} 

Condition~\ref{itm:B1} implies that $\frac{ng}{k}\le \frac{q^{b/k}-1}{b}$, and thus, since $g\le k-1$, it follows that any integer~$b$ such that $b\ge k(\log_q n+\log_q b)$ suffices. Further, Condition~\ref{itm:B1} implies that
\[
B=\Omega\left(k^2(\log_q(n)+\log_q(b))^2 \right),
\]
and hence it is trivially satisfied in many distributed storage systems.

\subsection{Construction}\label{section:MSRconstruction}
Similar to~\cite{PMcodes}, given a file~$x\in\bF_q^B$, arrange its symbols in the upper triangle of two square matrices $S_1,S_2$ of dimensions $\frac{b(k-1)}{k}\times \frac{b(k-1)}{k}$ over~$\bF_q$, complete the lower triangle of $S_1,S_2$ to obtain symmetry, and define
\begin{align*}
X\triangleq\begin{pmatrix}
S_1 \\ S_2
\end{pmatrix}.
\end{align*}

Next, a set of integers~$i_1,\ldots,i_n$ in the range $\{0,\ldots,\frac{q^{b/k}-1}{g}-1\}$ is chosen such that no two reside in the same $q$-cyclotomic coset modulo~$\frac{q^{b/k}-1}{g}$. This choice is enabled by the following lemma.

\begin{lemma}\label{lemma:SizeOfCosets}
	The size of $q$-cyclotomic cosets modulo~$\frac{q^{b/k}-1}{g}$ is at most~$b/k$.
\end{lemma}

\begin{proof}
	According to~\cite[Th.~4.1.4, p.~123]{HuffmanAndPless}, for any~$m$, the size of any $q$-cyclotomic coset modulo~$m$ is a divisor of $\ord_m(q)$, where $\ord_m(q)$ is the smallest integer~$t$ such that~$q^t = 1~(\bmod m)$. Since clearly, $\frac{q^{b/k}-1}{g}\vert q^{b/k}-1$, it follows that $q^{b/k}= 1~(\bmod \frac{q^{b/k}-1}{g})$, which implies that $\ord_{(q^{b/k}-1)/g}(q)$ is at most~$b/k$, and the claim follows.
\end{proof}

Lemma~\ref{lemma:SizeOfCosets} implies that there are at least~$\frac{q^{b/k}-1}{g\cdot (b/k)}$ different $q$-cyclotomic cosets modulo~$\frac{q^{b/k}-1}{g}$, which enables the choice of $i_1,\ldots,i_n$ by Condition~\ref{itm:B1}. Notice that the choice of $i_1,\ldots,i_n$ is possible using a simple algorithm, which maintains a list of feasible elements, iteratively picks an arbitrary element as the next $i_j$, and removes its coset from the list.

Let~$P$ be a companion matrix of any primitive polynomial of degree~$\frac{b}{k}$ over~$\bF_q$, and let

\begin{align*}
\Phi &\triangleq
\begin{pmatrix}
I & P^{i_1}   & \cdots & P^{i_1(k-2)}\\
I & P^{i_2} & \cdots & P^{i_2(k-2)}\\
\vdots & \vdots & \vdots & \vdots \\
I & P^{i_n} & \cdots & P^{i_n(k-2)}\\
\end{pmatrix}\in\bF_{q}^{\frac{bn}{k}\times\frac{b(k-1)}{k}} & \Lambda \triangleq \begin{pmatrix}
P^{i_1(k-1)} & ~ & ~ &~\\
~     & P^{i_2(k-1)} & ~ &~\\
~ & ~ & \ddots & \\
~ & ~ & ~ & P^{i_n(k-1)} \\
\end{pmatrix} \in\bF_{q}^{\frac{bn}{k}\times\frac{bn}{k}} .
\end{align*}

Define the~$\frac{bn}{k}\times\frac{bd}{k}$ encoding matrix over~$\bF_q$ as $M\triangleq \begin{pmatrix}\Phi&\Lambda\Phi\end{pmatrix}$ and notice that~$M$ is a block-Vandermonde matrix. Moreover, according to Lemma~\ref{lemma:ExtensionFieldRepresentation}, Lemma~\ref{lemma:ThetaInv}, and the choice of~$i_1,\ldots,i_n$, it follows from the properties of Vandermonde matrices in~$\bF_{q^{b/k}}^{n\times d}$ that any $\frac{bd}{k}\times \frac{bd}{k}$ block submatrix\footnote{That is, a submatrix which consists of complete blocks.} of~$M$ is invertible. Similarly, every $\frac{b(k-1)}{k}\times \frac{b(k-1)}{k}$ block submatrix of~$\Phi$ is invertible. Let $M_i$ be the $i$-th block-row of $M$, and store $M_i\cdot X$ in node~$i$. By the definition of the corresponding matrices, we have that~$\alpha=\frac{b^2(k-1)}{k^2}$.

\begin{remark} \label{remark:MSRk=b}
	For $k=b$ this code is a special case of an PM-MSR code~\cite[Sec.~V]{PMcodes}, and in which case condition~\ref{itm:B1} implies that $q\ge n\cdot\gcd(k-1,q-1)+1$. Hence, the advantage of our techniques exists not only for $b>k$, unlike Remark~\ref{remark:MBRk=b}. This improvement also follows from~\cite[Eq.~(37)]{MSRq>n}.
\end{remark}

\begin{theorem}\label{theorem:MSRrepair}
	In the above code, exact repair of any failed node may be achieved by downloading~$\beta\triangleq \frac{b^2}{k^2}$ field elements from any~$d$ of the remaining nodes.
\end{theorem}

\begin{proof}
	Assume that node~$\ell$ failed and $D=\{j_1,\ldots,j_d\}$ is a subset of~$[n]$ of size~$d$ such that~$\ell\notin D$. Let~$\Phi_\ell$ be the $\ell$-th block-row of~$\Phi$, and notice that node~$\ell$ stored 
	\begin{align}\label{eqn:MSRrepair}
	M_\ell X=\begin{pmatrix}\Phi_\ell& P^{i_\ell(k-1)}\Phi_\ell \end{pmatrix}\cdot X= \Phi_\ell S_1+P^{i_\ell(k-1)}\Phi_\ell S_2.
	\end{align}
	
	To repair node~$\ell$, every node $j_t\in D$ computes $M_{j_t}X\Phi_\ell^\top$, which is a~$\frac{b}{k}\times \frac{b}{k}$ matrix over~$\bF_q$, and sends it to the newcomer. The newcomer obtains
	\begin{align*}
	\begin{pmatrix}
	M_{j_1}X\Phi_\ell^\top\\M_{j_2}X\Phi_\ell^\top\\\vdots\\M_{j_d}X\Phi_\ell^\top
	\end{pmatrix}=
	\begin{pmatrix}
	M_{j_1}\\M_{j_2}\\\vdots\\M_{j_d}
	\end{pmatrix}\cdot X \cdot \Phi_\ell^\top\triangleq M_DX\Phi_\ell^\top.
	\end{align*}

Since $M_D$ can be seen as $\Theta(M'_D)$ for some full rank Vandermonde matrix~$M'_D\in\bF_{q^{b/k}}^{d\times d}$, it follows from Lemma~\ref{lemma:ThetaInv} that $M_D$ is invertible, and hence the newcomer may obtain
\[
\left(X\Phi_\ell^\top\right)^\top=\Phi_\ell\cdot\begin{pmatrix}
S_1 & S_2
\end{pmatrix},
\]
and restore~$M_\ell X$ by~\eqref{eqn:MSRrepair}.
\end{proof}

\begin{theorem}\label{theorem:MSRreconstruction}
	In the above code, reconstruction may be achieved by downloading $\alpha=\frac{b^2(k-1)}{k^2}$ field elements per node from any~$k$ nodes.
\end{theorem}

\begin{proof}
	Let $K=\{j_1,\ldots,j_k\}$ be a subset of~$[n]$ of size~$k$, and download $M_{j_i}X$ from node~$j_i$ for each $j_i\in K$. The data collector obtains
	\begin{align*}
	M_KX=\Phi_K S_1 + \Lambda_K\Phi_K S_2,
	\end{align*}
	where $\Lambda_K$ and $\Phi_K$ are the row-submatrices of~$\Lambda$ and~$\Phi$ which consist of the block-rows which are indexed by~$K$. By multiplying from the right by~$\Phi_K^\top$, the data collector obtains
	\begin{align*}
	\Gamma\triangleq \Phi_KS_1\Phi_K^\top+\Lambda_K\Phi_KS_2\Phi_K^\top\triangleq W+\Lambda_K Q,
	\end{align*}
	where $W$ and $Q$ are symmetric matrices. For~$s\in[k]$ denote~$i_{j_s}$ by~$\ell_s$, and notice that for distinct~$s$ and~$t$ in~$[k]$, 
	\begin{align}\label{eqn:MSRreconstruction1}
	\subm{\Gamma}_{s,t} &= \llbracket W \rrbracket _{s,t}+P^{\ell_s(k-1)}\llbracket Q \rrbracket_{s,t}\\
	\nonumber\llbracket \Gamma \rrbracket _{t,s} &= \llbracket W \rrbracket _{t,s}+P^{\ell_t(k-1)}\llbracket Q \rrbracket_{t,s}\\
   									\nonumber&= \llbracket W \rrbracket _{s,t}^{\top}+P^{\ell_t(k-1)}\llbracket Q \rrbracket_{s,t}^\top\\
   	\llbracket \Gamma \rrbracket _{t,s}^\top &= \llbracket W \rrbracket _{s,t}+\llbracket Q \label{eqn:MSRreconstruction2}\rrbracket_{s,t}\cdot\left(P^{\ell_t(k-1)}\right)^\top.
	\end{align}
	Thus, by subtracting~\eqref{eqn:MSRreconstruction2} from~\eqref{eqn:MSRreconstruction1} we have that
	\begin{align*}
	\subm{\Gamma}_{s,t}-\subm{\Gamma}_{t,s}^\top&=P^{\ell_s(k-1)}\llbracket Q \rrbracket_{s,t}-\llbracket Q\rrbracket_{s,t} \left(P^{\ell_t(k-1)}\right)^\top\\
	\left(\subm{\Gamma}_{s,t}-\subm{\Gamma}_{t,s}^\top\right)\cdot \left(P^{-\ell_t(k-1)}\right)^{\top}&=P^{\ell_s(k-1)}\llbracket Q \rrbracket_{s,t}\left(P^{-\ell_t(k-1)}\right)^{\top}-\llbracket Q\rrbracket_{s,t}.
	\end{align*}
	Now, it follows from Lemma~\ref{lemma:Vectorizing} that vectorizing both sides of this equation results in
	\begin{align}\label{eqn:MSRreconstruction4}
	\left(P^{-\ell_t(k-1)}\otimes P^{\ell_s(k-1)}-I_{\frac{b^2}{k^2}}\right)\cdot \vect(\subm{Q}_{s,t})=\vect\left(\left(\subm{\Gamma}_{s,t}-\subm{\Gamma}_{t,s}^\top\right)\cdot\left(P^{-\ell_t(k-1)}\right)^{\top}\right),
	\end{align}
	which may be seen as a linear system of equations whose variables are the unknown entries of~$\subm{Q}_{s,t}$. According to Lemma~\ref{lemma:eigenvalue1}, this equation has a unique solution if and only if~$1$ is not an eigenvalue of $P^{-\ell_t(k-1)}\otimes P^{\ell_s(k-1)}$. 
	
	Since the characteristic polynomial of any companion matrix is its corresponding polynomial, and since for~$P$ this polynomial is primitive, the eigenvalues of~$P$ are $\gamma,\gamma^q,\ldots,\gamma^{q^{b/k-1}}$, where~$\gamma$ is some primitive element in~$\bF_{q^{b/k}}$~\cite[Th.~4.1.1, p.~123]{HuffmanAndPless}. Therefore, the eigenvalues of~$P^{\ell_s(k-1)}$ are $\gamma^{\ell_s(k-1)},\gamma^{\ell_s(k-1)q},\ldots,\gamma^{\ell_s(k-1)q^{b/k-1}}$, the eigenvalues of~$P^{-\ell_t(k-1)}$ are $\gamma^{-\ell_t(k-1)},\gamma^{-\ell_t(k-1)q},\ldots,\gamma^{-\ell_t(k-1)q^{b/k-1}}$, and by Lemma~\ref{lemma:KroneckerEigenvalues}, the eigenvalues of $P^{-\ell_t(k-1)}\otimes P^{\ell_s(k-1)}$ are
	\begin{align*}
	\Delta\triangleq\left\{\gamma^{\ell_s(k-1)q^e-\ell_t(k-1)q^h }~\vert~ e,h\in\left\{0,1,\ldots,b/k-1\right\}\right\}.
	\end{align*}
	If~$1\in\Delta$, it follows that there exist~$e$ and~$h$ in~$\{0,1,\ldots,b/k-1\}$ such that 
	\begin{align*}
	\gamma^{\ell_s(k-1)q^e-\ell_t(k-1)q^h }=1,
	\end{align*}
	which implies that $\ell_s(k-1)q^e=\ell_t(k-1)q^h~(\bmod q^{b/k}-1)$. Therefore, there exists an integer~$t$ such that
	\begin{align*}
	\ell_s(k-1)q^e &= \ell_t(k-1)q^h + t(q^{b/k}-1)\\
	\ell_sq^e\cdot \frac{k-1}{g} &= \ell_tq^h\cdot\frac{k-1}{g} + t\cdot\frac{q^{b/k}-1}{g},
	\end{align*}
	and thus,
	\begin{align}\label{eqn:MSRreconstruction3}
	\ell_sq^e\cdot \frac{k-1}{g} &= \ell_tq^h\cdot\frac{k-1}{g} \left(\bmod \frac{q^{b/k}-1}{g}\right).
	\end{align}

	Since clearly, $\gcd(\frac{k-1}{g},\frac{q^{b/k}-1}{g})=1$, it follows that~$\frac{k-1}{g}$ is invertible modulo~$\frac{q^{b/k}-1}{g}$. Therefore,~\eqref{eqn:MSRreconstruction3} implies that $\ell_sq^e = \ell_tq^h (\bmod \frac{q^{b/k}-1}{g})$. Since $\gcd(q,\frac{q^{b/k}-1}{g})=1$, it follows that $q$ is invertible modulo~$\frac{q^{b/k}-1}{g}$. Hence, we have that $\ell_s = \ell_tq^{h-e} (\bmod \frac{q^{b/k}-1}{g})$ if $h\ge e$ and $\ell_sq^{e-h} = \ell_t (\bmod \frac{q^{b/k}-1}{g})$ if $h<e$. Either way, it follows that~$\ell_t$ and $\ell_s$, which are notations for~$i_{j_t}$ and~$i_{j_s}$, respectively, are in the same $q$-cyclotomic coset modulo~$\frac{q^{b/k}-1}{g}$, a contradiction to the choice of $i_1,\ldots,i_n$. Therefore, $1\notin\Delta$, which implies that~\eqref{eqn:MSRreconstruction4} is solvable, and the data collector may obtain $\subm{Q}_{s,t}$ and $\subm{W}_{s,t}$ for all distinct~$s$ and~$t$ in~$[k]$.
	
	Having this information, the data collector may consider the~$i$-th block-row of~$Q$, excluding the diagonal element,
	\begin{align*}
	\Phi_iS_2\begin{pmatrix}
	\Phi_1^\top & \cdots & \Phi_{i-1}^\top & \Phi_{i+1}^\top & \cdots & \Phi_k^\top
	\end{pmatrix},
	\end{align*}
	in which the matrix on the right is invertible by construction, and by Lemma~\ref{lemma:ThetaInv}. Hence, the data collector obtains $\Phi_1S_2,\ldots,\Phi_kS_2$, out of which any $k-1$ may once again be used to extract~$S_2$ by the same argument. Clearly, $S_1$ may be obtained similarly from the submatrices $\subm{W}_{s,t}$. 
\end{proof}

Note that in the above code $B=\frac{b^2(k-1)}{k}\left(1-\frac{1}{k}+\frac{1}{b}\right)$, and $\alpha k=\frac{b^2(k-1)}{k}$. Thus, the construction in this section \textit{does not} provide an MSR code. However, $\frac{B}{\alpha k}\overset{k\to\infty}{\longrightarrow}1$, and thus the cut-set bound is achieved asymptotically. A detailed comparison with PM-MSR codes and numerical examples appear in the Subsection~\ref{section:MSRasymptotic}.

This construction can be used to obtain NMSR codes for $d>2k-2$ in a recursive manner. 
By following a very similar outline to that of~\cite[Th.~6]{PMcodes}, we have the following.

\begin{theorem}\label{theorem:MSRd>2k-2}
	If~there exists an $(n',k',d',B',q,\alpha,\beta)$ regenerating code~$\cC'$ such that $\frac{B'}{\alpha k'}\overset{k'\to\infty}{\longrightarrow}1$, then there exists a $(n=n'-1,k=k'-1,d=d'-1,B'=B-\alpha,q,\alpha,\beta)$  regenerating code~$\cC$ such that $\frac{B}{\alpha k}\overset{k\to\infty}{\longrightarrow}1$.
\end{theorem}

\begin{proof}
	Without loss of generality assume that~$\cC'$ is systematic, and let~$\cC$ be the code which results from \textit{puncturing} the first systematic node of~$\cC'$. It follows from the properties of~$\cC'$ that~$\cC$ is a code with~$n=n'-1$ nodes, in which any $d=d'-1$ nodes can be used for repair, and any $k=k'-1$ nodes may be used for reconstruction. Moreover, $B=B'-\alpha$, and
	\begin{align*}
	\frac{B}{\alpha k}=\frac{B'-\alpha}{\alpha(k'-1)}=\frac{B'}{\alpha k'}\cdot \frac{k'}{k'-1}-\frac{1}{k'-1}\overset{k\to\infty}{\longrightarrow}1
	\end{align*}
\end{proof}

Notice that in Theorem~\ref{theorem:MSRd>2k-2}, if~$d'=ik'+j$ then $d=ik+j+(i-1)$. Hence, given the construction for~$d=2k-2$, one may obtain NMSR codes for larger values of~$d$. Moreover, it is evident that $\frac{B'}{\alpha k'}\le 1$, $\frac{k'}{k'-1}>1$, and that~$\frac{1}{k'-1}$ is negligible as~$k$ grows. Hence, the proof of Theorem~\ref{theorem:MSRd>2k-2} implies that $\frac{B}{\alpha k}$ tends to~$1$ \textit{faster} than $\frac{B'}{\alpha k'}$ does.

\subsection{The proximity of NMSR codes to MSR codes}\label{section:MSRasymptotic}
In this subsection the construction from Subsection~\ref{section:MSRconstruction} is compared with PM-MSR codes for the case $d=2k-2$. A comparison for the case $d>2k-2$ will appear in future versions of this paper. Following the reasoning which is described in Subsection~\ref{section:MBRconstruction}, the codes are compared over fields of even characteristic. That is, our codes are considered with~$q=2$, and since PM-MSR codes require $q\ge n(k-1)$, they are considered with~$q=2^{\ceil{\log(n(k-1))}}$. 

Similar to Definition~\ref{definition:EPM-MBR} and its subsequent discussion, EPM-MSR codes may also be defined. Note that a comparable loss of rate is apparent, not only due to the redundant representation, but also due to the symmetry which is required from the submatrices on the main diagonals of $S_1$ and~$S_2$.

The codes PM-MSR and EPM-MSR are concatenated to themselves in order to obtain the same $n,k,d,\alpha,$ and $\beta$, and only then the resulting file size and code rate are compared. The comparison for general parameters appears in Table~\ref{table:GeneralParametersMSR}, in which the values of $\alpha,\beta$, and $B$ are given in bits. Note that as in Subsection~\ref{section:MBRasymptotic}, the value of~$B$ for EPM-MSR is the number of \textit{information} bits, rather than the number of bits in the redundant representation. Further, numerical examples are given in Table~\ref{table:SpecificParamtersMSR}. Notice that it is possible to reduce the field size of PM-MSR codes in some cases~\cite{MSRq>n}. Yet, we compare our NMSR codes to PM-MSR for simplicity and generality.

\renewcommand{\arraystretch}{1.35}
\begin{table}[h]
	\centering
	\begin{tabular}{|c|c|c|c|}
		\cline{2-4}
		\multicolumn{1}{c}{~} &\multicolumn{1}{|c|}{NMSR} & \specialcell[c]{PM-MSR \\ concatenated $\frac{b^2}{\ceil{\log (n(k-1))}k^2}$ times}&\specialcell[c]{EPM-MSR \\ concatenated $\frac{b^2}{\ceil{\log (n(k-1))}^2k^2}$ times}\\ \hhline{-===}
		$q$ & $2$ & $2^{\ceil{\log (n(k-1))}}$&2\\ \hline
		$\beta$ & $\frac{b^2}{k^2}$  & \specialcell[c]{$\frac{b^2}{\ceil{\log (n(k-1))}k^2}$ field elements \\in vector form,\\i.e., $\frac{b^2}{k^2}$ bits.} & \specialcell[c]{$\frac{b^2}{\ceil{\log (n(k-1))}^2k^2}$ field elements\\in matrix form,\\i.e., $\frac{b^2}{k^2}$ bits.}\\ \hline
		$\alpha$& $\frac{b^2}{k^2}\cdot (k-1)$  & $\frac{b^2}{k^2}\cdot (k-1)$ & $\frac{b^2}{k^2}\cdot (k-1)$ \\ \hline
		$B$ & $\frac{b^2(k-1)}{k}\left(1-\frac{1}{k}+\frac{1}{b}\right)$ & $\frac{b^2(k-1)}{k}$ & $\frac{b^2(k-1)}{k\ceil{\log (n(k-1))}}$\\ \hline
		Rate & $\frac{k}{n}\left(1-\frac{1}{k}+\frac{1}{b}\right)$ & $\frac{k}{n}$ & $\frac{k}{n\ceil{\log(n(k-1))}}$\\ \hline
	\end{tabular}
	\caption{A comparison of parameters between our MSR codes (Subsection~\ref{section:MSRconstruction}) and the PM-MSR codes~\cite[Sec.~V]{PMcodes} for general $n,k,d=2k-2$.}\label{table:GeneralParametersMSR}
\end{table}
\renewcommand{\arraystretch}{1}

\begin{table}[h]
	\centering
		\begin{tabular}{|c|c|c|c|c|c|c|c|c|c|}
			\cline{2-10}
			\multicolumn{1}{c}{~}     & \multicolumn{1}{|c|}{$n$} & $k$ & $d$ & $\alpha$ & $\beta$ & $b$ & $q$ & $B$ & Rate \\\hhline{-=========}
			 NMSR  &\multirow{3}{*}{20}&\multirow{3}{*}{10}&\multirow{3}{*}{18}&\multirow{3}{*}{$45$KB}&\multirow{3}{*}{$5$KB}&\multirow{3}{*}{$200\cdot k$}&2&$405.225$KB&$0.45025$\\\cline{1-1}\cline{8-10}
			PM-MSR   &~					&~                 &~			      &~					 &~					   &~				  &256&$450$KB&$0.5$\\\cline{1-1}\cline{8-10}
			EPM-MSR   &~					&~                 &~			      &~					 &~					   &~				  &2&$56.25$KB&$0.0625$\\\hline\hline
			 NMSR  &\multirow{3}{*}{100}&\multirow{3}{*}{40}&\multirow{3}{*}{78}&\multirow{3}{*}{$7.02$MB}&\multirow{3}{*}{$0.18$MB}&\multirow{3}{*}{$1200\cdot k$}&2&$\approx 0.27$GB&$\approx 0.39$\\\cline{1-1}\cline{8-10}
			PM-MSR   &~					&~                 &~			      &~					 &~					   &~				  &4096&$0.28$GB&$0.4$\\\cline{1-1}\cline{8-10}
			EPM-MSR   &~					&~                 &~			      &~					 &~					   &~				  &2&$0.02$GB&$\approx 0.033$\\\hline\hline
			NMSR&\multirow{3}{*}{100}&\multirow{3}{*}{40}&\multirow{3}{*}{78}&\multirow{3}{*}{$175.5$MB}&\multirow{3}{*}{$4.5$MB}&\multirow{3}{*}{$6000\cdot k$}&2&$\approx 6.84$GB&$\approx 0.39$\\\cline{1-1}\cline{8-10}
			PM-MSR   &~					&~                 &~			      &~					 &~					   &~				  &4096&$7.02$GB&$0.4$\\\cline{1-1}\cline{8-10}
			EPM-MSR   &~					&~                 &~			      &~					 &~					   &~				  &2&$0.585$GB&$\approx 0.033$\\\hline\hline
			NMSR&\multirow{3}{*}{1000}&\multirow{3}{*}{400}&\multirow{3}{*}{798}&\multirow{3}{*}{$\approx 1800$KB}&\multirow{3}{*}{$\approx 4.5$KB}&\multirow{3}{*}{$190\cdot k$}&2&$\approx 0.718$GB&$\approx 0.39$\\\cline{1-1}\cline{8-10}
			PM-MSR   &~					&~                 &~			      &~					 &~					   &~				  &524288&$0.72$GB&$0.4$\\\cline{1-1}\cline{8-10}
			EPM-MSR   &~					&~                 &~			      &~					 &~					   &~				  &2&$37.9$MB&$\approx 0.02$\\\hline
		\end{tabular}
	\caption{A comparison of parameters between our code (Subsection~\ref{section:MSRconstruction}) and the PM-MSR code~\cite[Sec.~V]{PMcodes} for several common parameters $n,k,d=2k-2$.}\label{table:SpecificParamtersMSR}
\end{table}

\ifdefined\subspaces
\section{Connection to subspace codes}\label{section:SubspaceCodes}
Several connections between MBR codes and subspace codes were observed in the past~\cite{OggierProjective,Us,Hollman}. In these works, any storage node is associated with a subspace of dimension~$\alpha$ in $\bF_q^B$, and stores a projection of~$x$ on that subspace\footnote{That is, a multiplication of a spanning matrix of the subspace with the file~$x$.}. However, these approaches do not outperform PM-MBR codes, and are only applicable for restricted sets of parameters. In this section it is shown that the results of Section~\ref{section:MBRconstruction} may also be attained by a variant of the subspace code approach. This connection may seem as a generalization of the techniques in Section~\ref{section:MBRconstruction}, and besides being of independent interest, might be useful for applying our techniques in other scenarios. In what follows we present the essentials of subspace coding theory that are required for the construction which follows.

For any integer $m$, it is widely known that the vector space $\bF_q^m$ may be endowed with a multiplication operation using the vector space isomorphism $\bF_q^m\cong \bF_{q^m}$. Hence, every subspace of $\bF_q^m$ may be considered as a subspace of $\bF_{q^m}$. In what follows, $\grsmn{q}{m}{u}$ denotes the set of all $u$ dimensional subspaces of $\bF_{q^m}$. A (constant dimension) subspace code is a subset of $\grsmn{q}{m}{u}$, equipped with the subspace metric $d_S(U,V)=\dim U+\dim V-2\dim(U\cap V)$.

\textit{Cyclic subspace codes} are subspace codes which are closed under cyclic shifts, defined as follows. For a subspace $U\in\grsmn{q}{m}{u}$, and a nonzero field element $\alpha\in\bF_{q^m}^*\triangleq \bF_{q^m}\setminus\{0\}$, the cyclic shift~$\alpha U$ of~$U$ is the set $\alpha U\triangleq\{\alpha\cdot u\vert u\in U\}$. According to the distributive law, for any $\alpha\in\bF_{q^m}^*$, the set $\alpha U$ is a subspace of $\bF_{q^m}$ of the same dimension as $U$. 

A widely used notion in the theory of subspace codes is \textit{partial spreads} (see for example,~\cite[Sec.~III]{EtzionVardy}). A constant dimension subspace code~$\cC$ in $\grsmn{q}{m}{u}$ is called a partial spread if its minimum distance is~$2u$, i.e., each two subspaces intersect trivially. Furthermore, if each vector in~$\bF_{q^m}$ is contained in some subspace in~$\cC$ then it is called a \textit{spread}. It is known~\cite{EtzionVardy} that spreads exist if and only if~$u$ divides~$m$, and the set of all cyclic shifts of a subfield of~$\bF_{q^m}$ is a spread. In what follows, certain subsets of such spreads which satisfy an additional property are used. This property involves \textit{sets} of subspaces, in contrast with the distance property which only involves two, and is given by the following definitions. It is worth noting that partial spreads also appear in the work of~\cite{OggierProjective}.

\begin{definition}\label{definition:independentSubspaces}
	A set $\{U_i\}_{i=1}^k\subseteq\grsmn{q}{m}{u}$ is called an independent set if $\dim\left(\sum_{i=1}^{k}U_i\right)=\sum_{i=1}^{k}\dim U_i$.
\end{definition}

It is readily verified that a set $\{U_i\}_{i=1}^k$ is an independent set if and only if for all $j\in[k]$, we have that $U_j\cap\left(\sum_{i\in[k]\setminus\{j\}}U_i\right)=\{0\}$. Further, if $\{U_i\}_{i=1}^k\subseteq\grsmn{q}{b}{u}$ is an independent set then $u\cdot k\le m$.

\begin{definition}\label{definition:everyK}
	A set $\{U_i\}_{i=1}^n\subseteq\grsmn{q}{m}{u}$ is called an every-$k$ independent set if for every $T\subseteq[n]$, $|T|=k$ we have that $\{U_t\}_{t\in T}$ is an independent set.
\end{definition}

\begin{remark}\label{remark:Every1Indpendent}
For $u=1$ and $m=k$, an every-$k$ independent set may easily be obtained by considering any $[n,k]$ MDS code over $\bF_q$, and taking as $\{U_i\}_{i=1}^n$ the linear spans of the columns of its generator matrix~\cite[p.~119, Ex.~4.1]{Ronny'sBook}. 
Further, it can be shown (see Theorem~\ref{theorem:MDSEquivalence} in the Appendix) that for an integer~$b$ which is a multiple of~$k$, if an every-$k$ independent set $\{U_i\}_{i=1}^n\subseteq \grsmn{q}{b}{b/k}$ exists, then there exists an $[n,k]$ MDS code over $\bF_{q^{b/k}}$ which is linear over~$\bF_q$. Therefore, the notions of every-$k$ independent sets are equivalent to linear MDS codes up to a certain extent. 
\end{remark}

The parameters $q,n,k,d$ in this section are required to satisfy the conditions~\ref{itm:A1} and~\ref{itm:A2} from Section~\ref{section:MBR}. In addition we also have that $B=\frac{b(b+1)}{2}+b^2\left(\frac{d}{k}-1\right)$, where for convenience we denote $c\triangleq \frac{db}{k}-b$, so that $B=\frac{b(b+1)}{2}+bc$. Notice that $c$ is chosen such that $\frac{b}{k}=\frac{b+c}{d}$. The techniques in this section require to fix a certain mapping from $\bF_{q^{b+c}}$ to $\bF_{q}^{b+c}$, which is done as follows.

Since $\frac{b}{k}=\frac{b+c}{d}$ it follows that $d$ divides $b+c$, and hence $\bF_{q^{b/k}}=\bF_{q^{(b+c)/d}}$ is a subfield of $\bF_{q^{b+c}}$. Hence, let $\cV\triangleq\{v_1,\ldots,v_{d}\}$ be a basis of $\bF_{q^{b+c}}$ over $\bF_{q^{(b+c)/d}}$, and let $\cU\triangleq\{u_1,\ldots,u_{(b+c)/d}\}$ be a basis of $\bF_{q^{(b+c)/d}}$ over $\bF_q$. Since $\cV$ and $\cU$ are bases, it follows that 
\begin{align}\label{eqn:BasisVU}
\cV\cU\triangleq\left\{v_i\cdot u_j\Big\vert 1\le i\le d, 1\le j\le \frac{b+c}{d}\right\}
\end{align}

is a basis of $\bF_{q^{b+c}}$ over $\bF_q$ (see Lemma~\ref{lemma:BasisVUproof} in the appendix). The basis~$\cV\cU$ is used to fix a mapping $\Phi$ of the elements of $\bF_{q^{b+c}}$ to the elements of $\bF_{q}^{b+c}$. Identify the indices $1,\ldots,b+c$ with \[(1,1),(1,2),\ldots,(1,\frac{b+c}{d}),(2,1),(2,2),\ldots,(d,\frac{b+c}{d}),\]
respectively, and for $w\in\bF_{q^{b+c}}$, $w=\sum_{i,j}w_{i,j}v_iu_j$, where $w_{i,j}\in\bF_q$ for all $1\le i\le d,~1\le j\le \frac{b+c}{d}$, define $\Phi(w)$ as the vector of length $b+c$ over $\bF_q$ which contains $w_{i,j}$ in entry $(i,j)$. Notice that the function~$\Phi$ is linear over $\bF_q$, i.e., every $w,w'\in\bF_{q^{b+c}}$ and every~$\lambda,\mu\in\bF_q$ satisfy $\Phi(\lambda w+\mu w')=\lambda\Phi(w)+\mu \Phi(w')$. To avoid cumbersome notation, we use $w$ instead of $\Phi(w)$ wherever it is clear from context.

This section is organized as follows. In Subsection~\ref{section:EverydConstruction} we construct an every-$d$ independent set, for whom an additional property is proved. This property is a consequence of the particular representation $\Phi$ of $\bF_{q^{b+c}}$ as vectors in $\bF_{q}^{b+c}$. This set is then used to construct NMBR codes for any $d\ge k$ over any field $\bF_q$ in Subsection~\ref{section:SubspaceConstruction}. The encoding matrix of this code is shown to posses a Vandermonde-like structure in subsection~\ref{section:SubspaceVandermonde}.

\subsection{Construction of an every-$d$ independent set in $\grsmn{q}{b+c}{\frac{b+c}{d}}$}\label{section:EverydConstruction}
Since $b\ge \log_q n \cdot k$, it follows that $q^{b/k}=q^{(b+c)/d}\ge n$, and hence there exists Vandermonde matrices 
\begin{align}\label{eqn:MatrixA'}
A'&\triangleq 
\begin{pmatrix}
1 & 1 & \cdots & 1\\
\gamma_1 & \gamma_2 & \cdots & \gamma_n\\
\gamma_1^2 & \gamma_2^2 & \cdots & \gamma_n^2\\
\vdots & \vdots & \ddots & \vdots \\
\gamma_1^{d-1} & \gamma_2^{d-1} & \cdots & \gamma_n^{d-1}\\
\end{pmatrix},
&A\triangleq 
\begin{pmatrix}
1 & 1 & \cdots & 1\\
\gamma_1 & \gamma_2 & \cdots & \gamma_n\\
\gamma_1^2 & \gamma_2^2 & \cdots & \gamma_n^2\\
\vdots & \vdots & \ddots & \vdots \\
\gamma_1^{k-1} & \gamma_2^{k-1} & \cdots & \gamma_n^{k-1}\\
\end{pmatrix}
\end{align}
over $\bF_{q^{(b+c)/d}}$, with $\gamma_i\ne \gamma_j$ for all distinct $i$ and $j$ in $[n]$, which implies that every $d\times d$ submatrix of $A'$ and every $k\times k$ submatrix of $A$ are invertible. Using the basis $\cV=\{v_1,\ldots,v_d\}$ which was mentioned above, for every $i\in[n]$ define $\beta_i\triangleq \sum_{j=1}^{d}\gamma_i^{j-1}v_j$, and $\beta_i'\triangleq \sum_{j=1}^{k}\gamma_i^{j-1}v_j$. The proof of the following lemma relies on the fact that any $d$ (resp. $k$) columns of $A'$ (resp. $A$) are linearly independent over $\bF_{q^{b/k}}=\bF_{q^{(b+c)/d}}$.

\begin{lemma}~\label{lemma:BetaIndependent}
	\begin{itemize}
		\item [A.\namedlabel{itm:Every-kElements}{A}] Every $k$ distinct elements in $\{\beta_i'\}_{i=1}^n$ are linearly independent over $\bF_{q^{b/k}}$.
		\item [B.\namedlabel{itm:Every-dElements}{B}] Every $d$ distinct elements in $\{\beta_i\}_{i=1}^n$ are linearly independent over $\bF_{q^{b/k}}$.
	\end{itemize}
\end{lemma}

\begin{proof}
	Since the set $\{v_1,\ldots,v_d\}$ is a basis of $\bF_{q^{b+c}}$ over $\bF_{q^{(b+c)/d}}$, it follows that the set $\{v_1,\ldots,v_k\}$ is an independent set over $\bF_{q^{(b+c)/d}}$. Hence, the proofs of Part~\ref{itm:Every-kElements} and Part~\ref{itm:Every-dElements} are similar, and we list below only the proof of Part~\ref{itm:Every-dElements}.
	
	Let $J=\{j_1,\ldots,j_d\}$ be a subset of $[n]$ of size $d$, and assume that $\sum_{i=1}^{d}d_i\beta_{j_i}=0$ for some coefficients $d_i\in\bF_{q^{(b+c)/d}}$. According to the definition of the $\beta_i$-s we have that
	\begin{align}
	\nonumber\sum_{i=1}^{d}d_i\cdot \beta_{j_i}&=\sum_{i=1}^{d}d_i\sum_{t=1}^{d}\gamma_{j_i}^{t-1}v_t = \sum_{i=1}^{d}\sum_{t=1}^{d}d_i\gamma_{j_i}^{t-1}v_t=\sum_{t=1}^{d}\sum_{i=1}^{d}d_i\gamma_{j_i}^{t-1}v_t\\
	~&=\sum_{t=1}^{d}v_t\sum_{i=1}^{d}d_i\gamma_{j_i}^{t-1}=0.\label{eqn:PartBProof}
	\end{align}
	Since $d_i$ and $\gamma_{j_i}^{t-1}$ are elements of $\bF_{q^{(b+c)/d}}$ for all $i$ and $t$, we have that~\eqref{eqn:PartBProof} is a linear combination of the independent set $\{v_i\}_{i=1}^d$ over $\bF_{q^{(b+c)/d}}$. Therefore, $\sum_{i=1}^{d}d_i\gamma_{j_i}^{t-1}=0$ for all $t\in[d]$, which implies that $\sum_{i=1}^{d}d_ic_{j_i}=0$, where $c_{j_i}$ is the $i$-th column of $A'$. Since $A'$ is a Vandermonde matrix, the column vectors $\{c_{j_i}\}_{i=1}^d$ are independent over $\bF_{q^{(b+c)/d}}$, and thus $d_i=0$ for all~$i$.
\end{proof}

For $i\in[n]$ let $V_i\triangleq \beta_i\bF_{q^{(b+c)/d}}$ and $V_i'\triangleq \beta_i'\bF_{q^{(b+c)/d}}$, and notice that $V_i,V_i'\in\grsmn{q}{b+c}{\frac{b+c}{d}}$.

\begin{lemma}~\label{lemma:VeverykdIndependent}
	\begin{itemize}
		\item [A.\namedlabel{itm:Every-kSubspaces}{A}] The set $\{V_i'\}_{i=1}^{n}$ is an every-$k$ independent set.
		\item [B.\namedlabel{itm:Every-dSubspaces}{B}] The set $\{V_i\}_{i=1}^{n}$ is an every-$d$ independent set.		
	\end{itemize}
\end{lemma}

\begin{proof}
	To prove Part~\ref{itm:Every-kSubspaces}, let $J=\{j_1,\ldots,j_k\}$ be a subset of $[n]$ of size $k$. Consider the function $f:\bF_{q^{b/k}}^k\to\sum_{i=1}^{k}V'_{j_i}$ which maps $(\ell_1,\ldots,\ell_k)$ to $\sum_{i=1}^{k}\beta'_{j_i}\ell_i$. It is readily verified that $\dim\left(\sum_{i=1}^{k}V'_{j_i}\right)=b$ if and only if~$f$ is injective. Assume for contradiction that~$f$ is not injective, i.e., there exist two different tuples of elements $(e_1,\ldots,e_k)\in\bF_{q^{b/k}}^k$ and $(f_1,\ldots,f_k)\in\bF_{q^{b/k}}^k$ such that
	\begin{align*}
	\sum_{i=1}^{k}\beta_{j_i}'e_i=\sum_{i=1}^{k}\beta_{j_i}'f_i.
	\end{align*}
	
	Therefore, $\sum_{i=1}^{k}\beta'_{j_i}(e_i-f_i)=0$, which by Lemma~\ref{lemma:BetaIndependent}, Part~\ref{itm:Every-kElements}, implies that $e_i=f_i$ for all $i$, a contradiction. This implies that  $\dim\left(\sum_{i=1}^{k}V'_{j_i}\right)=b=k\cdot\frac{b}{k}=\sum_{i=1}^{k}\dim V'_{j_i}$, as required.
	
	To prove Part~\ref{itm:Every-dSubspaces}, let $T=\{t_1,\ldots,t_d\}$ be a subset of $[n]$ of size $d$. Consider the function $g:\bF_{q^{(b+c)/d}}^d\to\sum_{i=1}^{d}V_{j_i}$ which maps $(\ell_1,\ldots,\ell_d)$ to $\sum_{i=1}^{d}\beta_{j_i}\ell_i$. It is readily verified that $\dim\left(\sum_{i=1}^{k}V_{j_i}\right)=b+c$ if and only if~$g$ is injective. Assume for contradiction that~$g$ is not injective, i.e., there exist two different tuples of elements $(e_1,\ldots,e_d)\in\bF_{q^{(b+c)/d}}^d$ and $(f_1,\ldots,f_d)\in\bF_{q^{(b+c)/d}}^d$ such that
	\begin{align*}
	\sum_{i=1}^{d}\beta_{j_i}e_i=\sum_{i=1}^{d}\beta_{j_i}f_i.
	\end{align*}
	
	Therefore, $\sum_{i=1}^{d}\beta_{j_i}(e_i-f_i)=0$, which by Lemma~\ref{lemma:BetaIndependent}, Part~\ref{itm:Every-dElements}, implies that $e_i=f_i$ for all $i$, a contradiction. Hence, $\dim\left(\sum_{i=1}^{d}V_{j_i}\right)=b+c=d\cdot\frac{b+c}{d}=\sum_{i=1}^{d}\dim V_{j_i}$.
\end{proof}

The following technical lemma is essential for the construction which follows. The proof of this lemma uses the function~$\Phi$, which maps an element of $\bF_{q^{b+c}}$ to its representation according to the basis $\cV\cU$~\eqref{eqn:BasisVU}, as defined in the beginning of Section~\ref{section:SubspaceCodes}.

\begin{lemma}\label{lemma:technical}
	For all $i\in[n]$, there exists a matrix $\cM_i\in\bF_{q}^{(b+c)/d \times b}$ such that $V_i=\left<\cM_i\vert \cN_i\right>$ for some matrix $\cN_i\in\bF_{q}^{(b+c)/d \times c}$ and $V_i' =\left<\cM_i\vert \bold{0}\right>$, where $\bold{0}$ is the $\frac{b+c}{d}\times c$ zero matrix.
\end{lemma}

\begin{proof}
	By definition, the subspace $V_i$ is spanned by the vectors $\{\beta_iu_j\}_{j=1}^{(b+c)/d}=\{\sum_{t=1}^{d}v_t\gamma_i^{t-1}u_j\}_{j=1}^{(b+c)/d}$, and similarly,  the subspace $V_i'$ is spanned by the vectors $\{\beta_i'u_j\}_{j=1}^{(b+c)/d}=\{\sum_{t=1}^{k}v_t\gamma_i^{t-1}u_j\}_{j=1}^{(b+c)/d}$. 

	Since the function~$\Phi$ is linear, it follows that for every $j\in[\frac{b+d}{d}]$ and every~$i\in[n]$, we have that
	\begin{align*}
	\Phi\left(\sum_{t=1}^{d}v_t\gamma_i^{t-1}u_j\right)=\sum_{t=1}^{d}\Phi\left(v_t\gamma_i^{t-1}u_j\right).
	\end{align*}
	
	Moreover, it follows from the definition of~$\Phi$ that for all $t\in[d]$, the vector $\Phi\left(v_t\gamma_i^{t-1}u_j\right)$ contains the representation of $\gamma_i^{t-1}u_j\in\bF_{q^{(b+c)/d}}$ by the basis $\cU$ in entries $(t,1),\ldots,(t,\frac{b+c}{d})$, and zero elsewhere. That is, if $\Psi:\bF_{q^{(b+c)/d}}\to\bF_{q}^{(b+c)/d}$ is the function which maps an element of~$\bF_{q^{(b+c)/d}}$ to its representation by $\cU$, and $\Psi(\gamma_i^{t-1}u_j)\triangleq (\eta_1,\ldots,\eta_{(b+c)/d})$, then
	\begin{align*}
	\Phi\left(v_t\gamma_i^{t-1}u_j\right)&=\Phi\left(v_t\sum_{s=1}^{(b+c)/d}\eta_su_s\right)\\
	~&= \Phi\left(\sum_{s=1}^{(b+c)/d}\eta_sv_tu_s\right)\\
	~&= (0,\ldots,0,\underbrace{\eta_1,\ldots,\eta_{(b+c)/d}}_\text{ \scriptsize$\begin{array}{cc}\text{indices}\\ (t,1),\ldots,(t,\frac{b+c}{d}) 	
		\end{array} $ },0,\ldots,0).
	\end{align*}
	
	Hence, for every $j\in[\frac{b+c}{d}]$ and every~$i\in[n]$ we have that
	\begin{align*}
	\Phi\left(\beta_iu_j\right)=\Phi\left(\sum_{t=1}^{d}v_t\gamma_i^{t-1}u_j\right)=\Psi\left(u_j\right)\circ \Psi\left(\gamma_iu_j\right)\circ\cdots\circ\Psi\left(\gamma_i^{d-1}u_j\right),
	\end{align*}
	where~$\circ$ denotes concatenation, and similarly,
	\begin{align*}
	\Phi\left(\beta_i'u_j\right)=\Phi\left(\sum_{t=1}^{k}v_t\gamma_i^{t-1}u_j\right)=\Psi\left(u_j\right)\circ \Psi\left(\gamma_iu_j\right)\circ\cdots\circ\Psi\left(\gamma_i^{k-1}u_j\right)\circ\underbrace{		\Psi(0)\circ\cdots\circ\Psi(0)}_\text{$d-k$ times}.
	\end{align*}

	Hence, for every~$i\in[n]$, the subspaces $V_i$ and $V_i'$ have the spanning matrices
	\begin{align*}
	V_i=\left<\left(
	\begin{array}{c|c|c|c}
	\Psi(u_1)        &\Psi(\gamma_i u_1)        &\cdots&\Psi(\gamma_i^{d-1}u_1)\\
	\Psi(u_2)        &\Psi(\gamma_i u_2)        &\cdots&\Psi(\gamma_i^{d-1}u_2)\\
	\vdots   	     &\vdots              		&\ddots&\vdots \\
	\Psi(u_{(b+c)/d})&\Psi(\gamma_i u_{(b+c)/d})&\cdots&\Psi(\gamma_i^{d-1}u_{(b+c)/d})
	\end{array}
	\right)\right>,~\text{and}\\
	V_i'=\left<\left(
	\begin{array}{c|c|c|c|c|c|c}
	\Psi(u_1)        &\Psi(\gamma_i u_1)        &\cdots&\Psi(\gamma_i^{k-1}u_1)&0&\cdots&0\\
	\Psi(u_2)        &\Psi(\gamma_i u_2)        &\cdots&\Psi(\gamma_i^{k-1}u_2)&0&\cdots&0\\
	\vdots   	     &\vdots                    &\ddots&\vdots &\vdots&\cdots&\vdots\\
	\Psi(u_{(b+c)/d})&\Psi(\gamma_i u_{(b+c)/d})&\cdots&\Psi(\gamma_i^{k-1}u_{(b+c)/d})&0&\cdots&0
	\end{array}
	\right)\right>.
	\end{align*}
\end{proof}

Thus, for any given set of $n$ distinct elements $\{\gamma_i\}_{i=1}^n$ in $\bF_{q^{(b+c)/d}}$, define the following matrices which will be used in the sequel. 

\begin{align}\label{eqn:Mi}
\cM_i\triangleq 
\begin{pmatrix}
	\begin{array}{c|c|c|c}
	\Psi(u_1)        &\Psi(\gamma_i u_1)        &\cdots&\Psi(\gamma_i^{k-1}u_1)\\
	\Psi(u_2)        &\Psi(\gamma_i u_2)        &\cdots&\Psi(\gamma_i^{k-1}u_2)\\
	\vdots   	     &\vdots                    &\ddots&\vdots \\
	\Psi(u_{(b+c)/d})&\Psi(\gamma_i u_{(b+c)/d})&\cdots&\Psi(\gamma_i^{k-1}u_{(b+c)/d})
	\end{array}
\end{pmatrix}\in \bF_q^{\frac{b}{k}\times b}
\end{align}

\begin{align}\label{eqn:Ni}
\cN_i\triangleq 
\begin{pmatrix}
	\begin{array}{c|c|c}
	\Psi(\gamma_i^{k} u_1)        &\cdots&\Psi(\gamma_i^{d-1}u_1)\\
	\Psi(\gamma_i^{k} u_2)        &\cdots&\Psi(\gamma_i^{d-1}u_2)\\
	\vdots                          &\ddots&\vdots\\
	\Psi(\gamma_i^{k} u_{(b+c)/d})&\cdots&\Psi(\gamma_i^{d-1}u_{(b+c)/d})
	\end{array}
\end{pmatrix}\in \bF_q^{\frac{b}{k}\times c}
\end{align}

\subsection{Nearly MBR codes from subspace codes}\label{section:SubspaceConstruction}
Arrange the symbols of the file $x\in \bF_q^B$ as in~\eqref{eqn:DataMatrix}. Encode $x$ as follows,
\begin{align*}
\begin{pmatrix}
\cM_1 & \cN_1 \\
\cM_2 & \cN_2 \\
\vdots & \vdots \\
\cM_n & \cN_n \\
\end{pmatrix} \cdot 
\begin{pmatrix}
S & T\\
T^\top & 0
\end{pmatrix}=
\begin{pmatrix}
\cM_1S+\cN_1T^\top & \cM_1 T\\
\cM_2S+\cN_2T^\top & \cM_2 T\\
\vdots & \vdots \\
\cM_nS+\cN_nT^\top & \cM_n T\\
\end{pmatrix},
\end{align*}
where $\cM_i\in\bF_q^{b/k\times b}$ and $\cN_i\in\bF_q^{b/k\times c}$ were defined in~\eqref{eqn:Mi} and~\eqref{eqn:Ni}, respectively. By the definition of the matrices $\cM_i,\cN_i,S,$ and $T$, we have that $\alpha=\frac{(b+c)^2}{d}=\frac{b^2}{k^2}\cdot d$.

\begin{theorem}\label{theorem:d>kRepair}
	In the above code, exact repair of any failed node may be achieved by downloading $\beta\triangleq \frac{(b+c)^2}{d^2}=\frac{b^2}{k^2}$ field elements from any $d$ of the remaining nodes.
\end{theorem}

\begin{proof}
	Assume that node $i$ failed, and $D=\{j_1,\ldots,j_d\}$ is a subset of $[n]$ of size $d$ such that $i\notin D$. To repair node $i$, every node $j_t\in J$ computes $(\cM_{j_t}\vert \cN_{j_t})X(\cM_i\vert \cN_i)^\top$, which is a $\frac{b+c}{d}\times\frac{b+c}{d}$ matrix over~$\bF_q$, and sends it to the newcomer. The newcomer obtains
	\begin{align*}
		\begin{pmatrix}
		(\cM_{j_1}\vert \cN_{j_1})X(\cM_i\vert \cN_i)^\top\\
		(\cM_{j_2}\vert \cN_{j_2})X(\cM_i\vert \cN_i)^\top\\
		\vdots\\
		(\cM_{j_d}\vert \cN_{j_d})X(\cM_i\vert \cN_i)^\top\\
		\end{pmatrix}=
		\begin{pmatrix}
		(\cM_{j_1}\vert \cN_{j_1})\\
		(\cM_{j_2}\vert \cN_{j_2})\\
		\vdots\\
		(\cM_{j_d}\vert \cN_{j_d})
		\end{pmatrix}\cdot X \cdot(\cM_i\vert \cN_i)^\top\triangleq (\cM\vert \cN)_DX (\cM_i\vert \cN_i)^\top.
	\end{align*}
	
	According to Lemma~\ref{lemma:technical} we have that $\left<(\cM_{j_i}\vert \cN_{j_i})\right>=V_{j_i}$, and by Lemma~\ref{lemma:VeverykdIndependent}, Part~\ref{itm:Every-dSubspaces}, we have that the row span of $(\cM\vert \cN)_D\in\bF_q^{(b+c)\times(b+c)}$ is~$\bF_q^{b+c}$, and hence it is invertible. Thus, the newcomer may obtain the missing data by multiplying by $(\cM\vert \cN)_D^{-1}$ and transposing.
\end{proof}

\begin{theorem}\label{theorem:d>kReconstruction}
	In the above code, reconstruction may be achieved by downloading $\alpha=\frac{(b+c)^2}{d}=\frac{b^2}{k^2}\cdot d$ field elements from any $k$ nodes. 
\end{theorem}

\begin{proof}
	Let $K=\{j_1,\ldots,j_k\}$ be a subset of $[n]$ of size $k$. By downloading the entire content $(\cM_{j_i}\vert \cN_{j_i})X$ from node $j_i$ for each $j_i\in J$, the data collector obtains
	\begin{align*}
	\begin{pmatrix}
	\cM_{j_1}S+\cN_{j_1}T^\top & \cM_{j_1} T\\
	\cM_{j_1}S+\cN_{j_2}T^\top & \cM_{j_2} T\\
	\vdots & \vdots \\
	\cM_{j_k}S+\cN_{j_k}T^\top & \cM_{j_k} T\\
	\end{pmatrix}\triangleq
	\begin{pmatrix}
	\cM_{j_1}S+\cN_{j_1}T^\top &\multirow{4}{*}{$\cM_KT$}\\
	\cM_{j_1}S+\cN_{j_2}T^\top & ~\\
	\vdots & ~\\
	\cM_{j_k}S+\cN_{j_k}T^\top & ~\\
	\end{pmatrix}.
	\end{align*}
	
	Recall that according to Lemma~\ref{lemma:technical}, we have that $\left<(\cM_{j_i}\vert \textbf{0})\right>=V_{j_i}'$. Since $\{V_{j_i}'\}_{i=1}^k$ is an independent set by Lemma~\ref{lemma:VeverykdIndependent}, Part~\ref{itm:Every-kSubspaces}, it follows that $\dim(\sum_{i=1}^{k}V_{j_i}' )=\sum_{i=1}^{k}\dim(V_{j_i}')=\frac{b}{k}\cdot k =b$, and hence, 
	\begin{align*}
	\rank 
	\begin{pmatrix}
	\cM_{j_1} & \bold{0}\\
	\cM_{j_1} & \bold{0}\\
	\vdots  & \vdots  \\
	\cM_{j_k} & \bold{0}\\
	\end{pmatrix}=b
	\end{align*}
	which implies that $\cM_K\in \bF_q^{b\times b}$ is invertible. Therefore, the data collector may extract the rightmost $c$ columns of his data, multiply by $\cM_K^{-1}$, and obtain $T$. Having $T$, the data collector may similarly obtain $S$ as well.
\end{proof}

Since $\beta=\frac{(b+c)^2}{d^2}$, $\alpha=\frac{(b+c)^2}{d}$, we have that $\beta d=\alpha$, and hence this code enables minimum bandwidth repair. Notice also that rate $\frac{B}{\alpha n}$ equals $\frac{k^2}{dn}\cdot(\frac{1}{2b}+\frac{d}{k}-\frac{1}{2})\approx\frac{2k-k^2/d}{2n}$. Hence, this construction is identical in parameters to the one given in Section~\ref{section:MBRconstruction}, and the analysis in Section~\ref{section:MBRasymptotic} applies to it as well. Yet, the subspace interpretation which is given in this section could possibly be implemented with alternative constructions of every-$d$ independent sets, which may be of independent mathematical interest.

\subsection{Vandermonde matrix structure}\label{section:SubspaceVandermonde}
The repair algorithm in the proof of Theorem~\ref{theorem:d>kRepair} relies on inverting a matrix of the form $(\cM\vert \cN)_{D}$ for some $d$-subset~$D$ of $[n]$, and the reconstruction algorithm in the proof of Theorem~\ref{theorem:d>kReconstruction} relies on inverting a matrix of the form $\cM_K$ for some $k$-subset $K$ of~$[n]$. These matrix posses a hidden Vandermonde structure which may be utilized for the inversion process by either an explicit formula~\cite[Sec.~1.2.3, Ex.~40]{Knuth} or a Reed-Solomon decoder~\cite[Ch.~6]{Ronny'sBook}. The Vandermonde structure of $\cM_K$ is presented, and the one of $(\cM\vert \cN)_{D}$ may be obtained similarly.

According to~\eqref{eqn:Mi}, for $K=\{j_1,\ldots,j_k\}$ the $b\times b$ matrix $\cM_K$ over $\bF_q$ is of the form

\begin{align}\label{eqn:MJ}
\cM_K=
\begin{pmatrix}
\begin{array}{c|c|c|c}
\Psi(u_1)        &\Psi(\gamma_{j_1} u_1)        &\cdots&\Psi(\gamma_{j_1}^{k-1}u_1)\\
\Psi(u_2)        &\Psi(\gamma_{j_1} u_2)        &\cdots&\Psi(\gamma_{j_1}^{k-1}u_2)\\
\vdots   	     &\vdots                    &\ddots&\vdots \\
\Psi(u_{(b+c)/d})&\Psi(\gamma_{j_1} u_{(b+c)/d})&\cdots&\Psi(\gamma_{j_1}^{k-1}u_{(b+c)/d})\\ \hline
\Psi(u_1)        &\Psi(\gamma_{j_2} u_1)        &\cdots&\Psi(\gamma_{j_2}^{k-1}u_1)\\
\Psi(u_2)        &\Psi(\gamma_{j_2} u_2)        &\cdots&\Psi(\gamma_{j_2}^{k-1}u_2)\\
\vdots   	     &\vdots                    &\ddots&\vdots \\
\Psi(u_{(b+c)/d})&\Psi(\gamma_{j_2} u_{(b+c)/d})&\cdots&\Psi(\gamma_{j_2}^{k-1}u_{(b+c)/d})\\ \hline
\vdots & \vdots & \vdots & \vdots \\ \hline
\Psi(u_1)        &\Psi(\gamma_{j_k} u_1)        &\cdots&\Psi(\gamma_{j_k}^{k-1}u_1)\\
\Psi(u_2)        &\Psi(\gamma_{j_k} u_2)        &\cdots&\Psi(\gamma_{j_k}^{k-1}u_2)\\
\vdots   	     &\vdots                    &\ddots&\vdots \\
\Psi(u_{(b+c)/d})&\Psi(\gamma_{j_k} u_{(b+c)/d})&\cdots&\Psi(\gamma_{j_k}^{k-1}u_{(b+c)/d})\\
\end{array}
\end{pmatrix},
\end{align}
where the function~$\Psi$ maps an element in~$\bF_{q^{b/k}}$ to its vector representation in~$\bF_q^{b/k}$ using the basis $\cU=\{u_1,\ldots,u_{b/k}\}$. Hence, by inverting~$\Psi$, the matrix $\cM_K$, which is a ${b\times b}$ matrix over~$\bF_q$, can be considered as a $b\times k$ matrix over $\bF_{q^{b/k}}$. In addition, for $i\in[\frac{b}{k}]$, if $M_{J_i}$ is the $k\times k$ matrix over $\bF_{q^{b/k}}$ which consists of all the $i$-th rows from each row-block in~\eqref{eqn:MJ}, then
\begin{align*}
\cM_{J_i}=u_i\cdot
\begin{pmatrix}
1 & \gamma_{j_1} & \gamma_{j_1}^2 & \cdots & \gamma_{j_1}^{k-1}\\
1 & \gamma_{j_2} & \gamma_{j_2}^2 & \cdots & \gamma_{j_2}^{k-1}\\
\vdots & \vdots & \vdots & \ddots & \vdots \\
1 & \gamma_{j_k} & \gamma_{j_k}^2 & \cdots & \gamma_{j_k}^{k-1}\\
\end{pmatrix}.
\end{align*}

Therefore, the matrix $\cM_K$ can be seen as a row-interleaving of $\frac{b}{k}$ constant multiples of a Vandermonde matrix over $\bF_{q^{b/k}}$.
\fi
\section{Discussion and future research}\label{section:Discussion}
In this paper, asymptotically optimal regenerating codes were introduced. These codes attain the cut-set bound asymptotically as the reconstruction degree~$k$ increases, and may be defined over any field if the file size is reasonably large. Further, these codes enjoy several properties which are inherited from product matrix codes, such as the fact that helper nodes do not need to know the identity of each other\footnote{Although this property is apparent in most regenerating codes constructions, some constructions do require otherwise, such as~\cite{DeterminantCoding}, and some of the work of~\cite{Us1}.}, and the ability to add an extra storage node without encoding the file anew.

It is evident from Table~\ref{table:SpecificParamters} and Table~\ref{table:SpecificParamtersMSR} that for~$q=2$, a small loss of code rate is apparent already for feasible values of~$k$, and clearly, similar results hold for larger~$q$ as well. Since large finite field arithmetics is often infeasible, our results contribute to the feasibility of storage codes.

The research of storage codes has gained a considerable amount of attention lately. In particular, the results of~\cite{PMcodes}, which inspired ours, was expanded and improved in few recent papers. For example,~\cite{DeterminantCoding} generalized the PM-MBR construction to achieve other points of the trade-off through minor matrices, and~\cite{BandwidthAdaptive} presented an MBR code which supports an arbitrary number of helper nodes in the repair process. Among the research directions we currently pursue are the application of the techniques from the current paper to the aforementioned works, as well as to high rate MSR constructions, and analyzing the encoding, decoding, repair, and reconstruction complexities of our codes in comparison with PM codes.

\section*{Acknowledgments}
The work of Netanel Raviv was supported in part by the Israeli Science Foundation (ISF), Jerusalem, Israel, under  Grant  no.~10/12, The IBM Ph.D. fellowship, and the Mitacs organization, under the Globalink Israel-Canada Innovation Initiative. The author would like to express his sincere gratitude to Prof.~Frank Kschischang, Prof.~Tuvi Etzion, and Prof. Itzhak Tamo for many insightful discussions.

\ifdefined\subspaces
\section*{Appendix A - Systematic encoding}
Similar to~\cite[Sec.~IV.B]{PMcodes}, the codes presented in Section~\ref{section:MBR} and Section~\ref{section:SubspaceCodes} have a systematic form. This form can be directly obtained by using a Cauchy matrix, without the need to apply an invertible linear transform on the encoding matrix. However, using this Cauchy matrix requires a modification of Condition~\ref{itm:A1}, as explained in the sequel.

Recall that the encoding matrix~$M$~\eqref{eqn:DataMatrix} has a block-Vandermonde structure. Alternatively, it can be obtained from a $k\times n$ Vandermonde matrix over~$\bF_{q^{b/k}}$ and applying the representation of $\bF_{q^{b/k}}$ as powers of a companion matrix (Lemma~\ref{lemma:ExtensionFieldRepresentation}). The essential property of the matrix~$M$ which is used throughout the paper is the fact that any~$d\times d$ matrix, and every $k\times k$ submatrix which consists of rows~$1$ through~$k$ is invertible. Hence, we may alternatively consider
\begin{align*}
A''\triangleq 
\begin{pmatrix}
I_k & E \\
0   & F
\end{pmatrix},
\end{align*}
where ${E\choose F}$ is a $d\times(n-k)$ Cauchy matrix over~$\bF_{q^{b/k}}$, and $I_k$ is the $k\times k$ identity matrix. Since every square submatrix of a Cauchy matrix is invertible, applying the matrix representation to~$A''$ will result in a \textit{systematic} version of the code. However, to have a $d\times (n-k)$ Cauchy matrix over $\bF_{q^{b/k}}$ one must update Condition~\ref{itm:A1} to $b\ge \log_q(n+d-k)\cdot k$.

Similarly, in Section~\ref{section:SubspaceCodes} we have that the crucial property of the matrix $A'$ (resp., the matrix~$A$) in~\eqref{eqn:MatrixA'} is the fact that any set of $d$ (resp.~$k$) columns in it are linearly independent. Hence, a nearly equivalent code may be obtained by considering the above matrix~$A''\in \bF_{q^{b/k}}^{d\times n}$. Further, the encoding matrices $M_i$ and $N_i$, given in~\eqref{eqn:Mi} and~\eqref{eqn:Ni}, are replaced by

\begin{align*}
\cM_i\triangleq 
\begin{pmatrix}
\begin{array}{c|c|c|c}
\Psi(a_{1,i}u_1)        &\Psi(a_{2,i} u_1)        &\cdots&\Psi(a_{k,i} u_1)\\
\Psi(a_{1,i}u_2)        &\Psi(a_{2,i} u_2)        &\cdots&\Psi(a_{k,i} u_2)\\
\vdots   	     			  &\vdots                   &\ddots&\vdots \\
\Psi(a_{1,i}u_{(b+c)/d})&\Psi(a_{2,i} u_{(b+c)/d})&\cdots&\Psi(a_{k,i}u_{(b+c)/d})
\end{array}
\end{pmatrix}
\end{align*}

\begin{align*}
\cN_i\triangleq 
\begin{pmatrix}
\begin{array}{c|c|c}
\Psi(a_{k+1,i} u_1)        &\cdots&\Psi(a_{d,i} u_1)\\
\Psi(a_{k+1,i} u_2)        &\cdots&\Psi(a_{d,i} u_2)\\
\vdots                     &\ddots&\vdots\\
\Psi(a_{k+1,i} u_{(b+c)/d})&\cdots&\Psi(a_{d,i} u_{(b+c)/d})
\end{array}
\end{pmatrix},
\end{align*}
where $a_{i,j}$ denotes the $(i,j)$-th entry of~$A''$. Consequently, we have that
\begin{align*}
\begin{pmatrix}
\cM_1 \\ \cM_2 \\ \vdots \\ \cM_k
\end{pmatrix} = I_b
\end{align*}
and $\cN_1=\cN_2=\ldots =\cN_k=0$, which implies that nodes $1$ through $k$ together hold the information $(S~T)$. Hence, by using $A''$ instead of $A'$, a systematic version of the code is obtained. Notice that Condition~\ref{itm:A1} is to be replaced by $b\ge \log_q (n+d-k)\cdot k$ as well.

\section*{Appendix B - Omitted proofs}
\begin{theorem}\label{theorem:MDSEquivalence}
	For any integers $b$ and $k$ such that $b|k$ and for any prime power $q$, if there exists an every-$k$ independent set $\{U_i\}_{i=1}^n\subseteq \grsmn{q}{b}{\frac{b}{k}}$, then there exists an $[n,k]$ MDS code over $\bF_{q^{b/k}}$ which is linear over~$\bF_q$.
\end{theorem}

\begin{proof}
	The required MDS code is defined by a function $f:\bF_{q^{b/k}}^k\to\bF_{q^{b/k}}^n$ such that for any $c\in\image{f}$, where~$\image{f}$ is the image of $f$, the preimage $f^{-1}(c)$ can be computed given any $k$ symbols of $c$. For each $i\in[n]$ let $M_i\in\bF_q^{(b/k)\times b}$ be a matrix such that $\left<M_i\right>=U_i$, fix any basis $\cU\triangleq \{u_1,\ldots,u_{b/k}\}$ of $\bF_{q^{b/k}}$ over $\bF_q$, and use~$\cU$ to consider a given $x\in\bF_{q^{b/k}}^k$ as a vector in $\bF_q^b$. For each $j\in[n]$ define $f(x)_j$ as $xM_j^\top$, where $\cU$ is once again used to represent a vector in $\bF_q^{b/k}$ as an element in $\bF_{q^{b/k}}$, that is, $f(x)_j=\sum_{i=1}^{b/k}(xM_j^\top)_iu_i$.
	
	To prove that $\image f$ is an MDS code, let $J\triangleq \{j_1,\ldots,j_k\}$ be a subset of $[n]$ of size $k$, and assume that $c_{j_1},\ldots,c_{j_k}$ are given for some $c\in\image{f}$. Since $\sum_{i=1}^{k}U_{j_i}=\bF_q^b$, and since $c_{j_i}=xM_{j_i}^\top$ for some $x\in\bF_q^b$ and for all $i\in[k]$, it follows that the matrix $\left(M_{j_1}^\top\vert M_{j_2}^\top\vert\cdots\vert M_{j_k}^\top\right)$ is invertible, and thus $x$ can be computed from $c_{j_1},\ldots,c_{j_k}$.
	
	To prove that $\image f$ is linear over $\bF_q$, for $x=(x_1,\ldots,x_k)\in\bF_{q^{b/k}}^k$ denote its representation in $\bF_q^b$ as $x=(x_{1,1},x_{1,2},\ldots,x_{1,b/k},x_{2,1},\ldots,x_{k,b/k})$. For all $x,y\in\bF_{q^{b/k}}^k$ and $j\in [n]$ we have that
	\begin{align*}
	f(x+y)_j=\sum_{i=1}^{b/k}((x+y)M_j^\top)_iu_i=\sum_{i=1}^{b/k}(xM_j^\top)_iu_i+\sum_{i=1}^{b/k}(yM_j^\top)_iu_i=f(x)_j+f(y)_j,
	\end{align*}
	and hence $f(x+y)=f(x)+f(y)$. Since for $a\in\bF_q$ we have that $f(ax)_j=\sum_{i=1}^{b/k}(axM_j^\top)_iu_i=af(x)$, the claim follows.
\end{proof}

\begin{lemma}\label{lemma:BasisVUproof}
	The set $\cV \cU$ from~\eqref{eqn:BasisVU} is a basis of $\bF_{q^{b+c}}$ over $\bF_q$.
\end{lemma}

\begin{proof}
	Assume that
	\begin{align*}
	\sum_{i=1}^{d}\sum_{j=1}^{\frac{b+c}{d}}w_{i,j}v_iu_j=0
	\end{align*}
	for some~$w_{i,j}$ in $\bF_q$, and notice that
	\begin{align}\label{eqn:VUproof1}
	\sum_{i=1}^{d}\sum_{j=1}^{\frac{b+c}{d}}w_{i,j}v_iu_j=\sum_{i=1}^{d}v_i\sum_{j=1}^{\frac{b+c}{d}}w_{i,j}u_j=0.
	\end{align}
	
	Since the elements of the basis~$\cU$ are in~$\bF_{q^{(b+c)/d}}$ and since $w_{i,j}$ are in~$\bF_q$ for all~$i$ and~$j$, it follows that $\sum_{j=1}^{(b+c)/d}w_{i,j}u_j$ is an element of $\bF_{q^{(b+c)/d}}$ for all~$i$. Therefore, since~$\cV$ is a basis of~$\bF_{q^{b+c}}$ over~$\bF_{q^{(b+c)/d}}$,~\eqref{eqn:VUproof1} implies that
	\begin{align}\label{eqn:VUproof2}
	\sum_{j=1}^{\frac{b+c}{d}}w_{i,j}u_j=0
	\end{align}
	for all~$i\in[d]$. Further, since~$\cU$ is a basis of~$\bF_{q^{(b+c)/d}}$ over~$\bF_q$, and since~$w_{i,j}\in\bF_q$,~\eqref{eqn:VUproof2} implies that~$w_{i,j}=0$ for all~$i\in[d]$ and for all~$j\in[\frac{b+c}{d}]$.
	\end{proof}
\fi
\end{document}